\newcommand{\bra}[1]{%
{\left\langle #1\right\vert}}
\newcommand{\ket}[1]{%
{\left\vert #1\right\rangle}}
\newcommand{\braket}[2]{%
{\left\langle #1  \vert  #2\right\rangle}}
\newcommand{\bracket}[3]{%
{\left\langle #1  \vert {#2} \vert #3 \right\rangle}}

\newcommand{\set}[1]{%
{\{ #1 \}}}
\newcommand{\abs}[1]{%
{\vert #1 \vert}}
\newcommand{\absq}[1]{%
{\vert #1 \vert^2}}
\newcommand{\proj}[1]{%
{\vert #1 \left\rangle \right\langle #1 \vert}}

\documentclass[12pt]{article}
\usepackage{amsmath,amsthm,amssymb}
\newtheorem{theorem}{Theorem}
\theoremstyle{plain}

\begin{document}
%
\title{SU(2) COHERENT STATE PATH INTEGRALS 
LABELED BY A FULL SET OF EULER ANGLES: 
BASIC FORMULATION}
\author{MASAO MATSUMOTO
\footnote{Class Preparation Office, 
Promotion of General Education and 
Liberal Arts Division, Student Affairs Department, 
Kyoto University, 
Yoshida Nihonmatsu-cho,
Sakyo-ku, Kyoto 606-8501, Japan
}\ \ 
\footnote{Present Address: 1-12-32 Kuzuha Asahi, Hirakara, 
Osaka 573-1111, Japan}
\\
{\em matumoto@i.h.kyoto-u.ac.jp}
\footnote
{Also available: masao-matsumoto@m-pine-tree.sakura.ne.jp }
}
%
%
\maketitle
\begin{abstract}
\noindent
We develop a basic formulation of the spin (SU(2)) coherent state 
path integrals based not on the conventional highest or 
lowest weight vectors but on arbitrary fiducial vectors. 
The coherent states, being defined on a 3-sphere, 
are specified by a full set of Euler angles. 
They are generally considered as states without classical analogues. 
The overcompleteness relation holds for the 
states, by which we obtain the time evolution of general systems 
in terms of the path integral representation; 
the resultant Lagrangian in the action has a monopole-type term 
{\em \`a la} Balachandran {\em etal.} as well as some additional terms, 
both of which depend on fiuducial vectors in a simple way. 
The process of the discrete path integrals to the continuous ones 
is clarified. 
Complex variable forms of the states and path integrals are also obtained. 
During the course of all steps, we emphasize the analogies and 
correspondences to the general canonical coherent states and  path integrals 
that we proposed some time ago. 
In this paper we concentrate on the basic formulation.
The physical applications as well as 
criteria in choosing fiducial vectors for real Lagrangians, 
in relation to fictitious monopoles and geometric phases, 
will be treated in subsequent papers separately. 
\end{abstract}

\renewcommand{\thefootnote}{\noindent \alph{footnote}}
\section{Introduction}
\label{sec:intro}
It has been approximately four decades since 
the ``coherent state (CS)'' for the Heisenberg-Weyl group, 
i.e. the ``canonical CS (CCS)'' was extended to wider classes \cite{Rad,Pera,Arec,BG}. 
During the period, the broader CS, together with the original one, 
have had a great influence on almost every branch of modern physics 
\cite{KlaSk,Perb,Perc,IKG,FKS}. 
\par
Since basic properties of CS are that they are continuous functions 
labeled by some parameters and that they compose 
``overcomplete'' sets \cite{KlaSk}, 
they provide a natural way to perform path integrations. 
Such ``coherent state path integrals (CSPI)”, 
i.e. path integrals (PI) via CS, have highly enriched the methods 
of PI with their physical applications 
\cite{KlaSk,IKG}. 
(In what follows each of the words ``CS'' and ``CSPI'' is used as a plural 
as well as a singular.)
\par 
As stated at the beginning, among all the CS 
the CCS is the original and the best-known CS that was
introduecd by Schr$\ddot{\rm o}$dinger 
\cite{Schroed}. 
The CCS is, in the light of quantum optics, generated by displacing, 
or driving, the vacuum, i.e. the zero photon state
 \cite{Glau}. 
From the viewpoint of CS in terms of unitary irreducible 
representations of Lie groups {\em \`a la} Perelomov 
 \cite{Pera,Perb,Perc}, 
the unitary operator is a displacing operator 
and a ``fiducial vector (FV)''
\footnote{
In \cite{NECSPI} and \cite{spinCSPI} we adopted the term ``starting vector'' 
which can be found in, e.g. p~14 of \cite{Perb}. 
The term seems well fit for the situation. 
However, we use ``fiducial vector'' in the present paper 
since it appears to be more employed in literature. 
See e.g. \cite{KlaSk}. 
We use also the word ``FV'' as a plural as well as a singular. 
} 
is the ground state or vacuum. 
\par
Some time ago we opened up the CSPI in terms of CCS evolving from an
arbitrary FV and investigated the associated geometric phases with an
application to quantum optics 
\cite{NECSPI}. 
Let us look back the results from the physical viewpoint 
concisely: 
First, we {\em set the generic CCS by displacing, or driving, 
not a usual vacuum, but an arbitrary superposition of photon number states}. 
If we take a single photon number state as a FV, we find that the CCS 
reduce to ``displaced number states (DNS)'' 
\cite{DNS}. 
So we may state that {\em the CCS with a generic FV is an arbitrary 
superposition of the DNS}. 
Second, using the general CCS we have performed CSPI which give a
completely general propagator including geometric phase terms 
corresponding to a quantum optical state that has no classical analogue; 
{\em The resultant action in the Lagrangian includes extra terms 
reflecting the entanglements between the coefficients of FV}. 
Third, particularly we investigated the geometric phase for DNS; 
And we found that the condition for the experimental detection of the phase 
had changed according to the $n$-dependence of DNS. 
Such CCS with a general FV may propose a clue to a universal language
for quantum optics especially when combined with SU(1,1) case.
\footnote
{``Displaced squeezed number states'' have been discussed 
in literature; See \cite{DNS} and references there in.} 
\par
Now, we realize that we have another CS which is of 
practical importance in a large variety of physical systems: 
That is the spin or SU(2) CS. 
It is one of the extensions of CS addressed at the beginning.
\footnote
{These two CS are intriguing also from the mathematical viewpoint 
since each of them is a typical example of CS for 
nilpotent Lie groups and semisimple ones respectively 
\cite{Pera}.
} 
Hence, following the achievements for the CCS above mentioned \cite{NECSPI}, 
we here endeavor to liberate spin CS from the conventional choice of FV,
$\ket{\Psi_0} = \ket{s, s}$ or $\ket{s, -s}$, 
and perform the path integration via the CS based on arbitrary FV; 
The SU(2) CS is, in this case, labeled by a full set of three Euler angles 
$(\phi, \theta, \psi)$; It is defined on a 3-sphere $S^3$. 
\par
We have several concrete reasons for doing such an extension. 
First, we know that DNS turned out to be non-classical quantum states 
that had various interesting properties 
 \cite{DNS}. 
So the analogous general spin CS may serve new non-classical quantum
states that have not been known. 
And the PI will provide examples 
of propagators for such states. 
Second, the usual CCS have been employed as logical gates in 
quantum computation (QC) 
\cite{Milburn}. 
And  moreover, the superpositions of DNS, which fall within the CCS 
with generic FV before mentioned 
\cite{NECSPI}, 
have already appeared in the context of QC 
 \cite{AAM-C}. 
Since spin systems as well as optical systems are probable candidates 
for exemplifying QC, spin CS with general FV may be also used in QC. 
Third, describing geometric phases, which has been one of the crucial topics 
in fundamental physics \cite{SW} in terms of CSPI requires the extension. 
Let us put it more concretely: 
Once elsewhere we investigated the geometric phases 
for a spin-$s$ particle under a magnetic field 
in the formalism of SU(2)CSPI with the conventional 
FV, i.e. $\ket{s, -s}$ \cite{KMa}. 
In consequence the results give the geometric phase 
of a monopole-type that corresponds merely 
to the adiabatic phase for the lowest eigenstate. 
However, it has been known that in the adiabatic phase 
for the same physical setting the strength of a fictitious monopole 
is proportional to the quantum number 
$m$ ($ m = -s, - s + 1, \cdots, s$) of the adiabatic state 
 \cite{Berry}. 
We cannot treat the corresponding case by the conventional SU(2)CSPI. 
Therefore the usual SU(2)CSPI is clearly unsatisfactory; 
And we had better let CS and CSPI prepare room also 
for the general cases which are reduced to any $m$th eigenstate 
in the adiabatic limit. 
Thus physics actually needs some extension of spin CSPI. 
Fourth, now geometric phases have been employed in QC 
 \cite{Sjoq}; 
And, as stated in the third reason, geometric phases are 
closely related to CSPI. 
Appreciating both areas it seems that we had better 
prepare wider CSPI and FV also for QC. 
Fifth, apart from the third reason, such a formalism 
of spin CSPI involving arbitrary FV 
may consequently shed a new light in understanding monopoles themselves 
in turn. 
For we have already known that the conventional spin CSPI provide a
mathematical description of monopoles naturally 
\cite{KMa}. 
And the description is common to real and fictitious monopoles. 
Hence it is possible that considering general FV in spin CSPI helps us 
to understand monopoles deeper.  
Finally, since the usual spin CS tends to the usual CCS 
in the high spin limit \cite{Rad}, 
we are naturally led to seek the spin CS and CSPI that are contracted 
to the CCS and CSPI with arbitrary FV described in 
\cite{NECSPI}. 
For the above several reasons we take a general FV in this paper. 
\par
We may grasp the main results by three theorems. 
The first one shows that we have the overcompleteness 
relation, or the resolution of unity, for the spin CS with 
an arbitrary FV. 
Using the result we obtain the second one as follows: 
The form of the generic Lagrangian for the SU(2)CSPI is 
(\ref{eqn:Lagpol}). 
As the Lagrangian for the usual CSPI, it consists of two parts: 
The topological term related to geometric phases 
and the dynamical one originating from a Hamiltonian. 
However, the contents are quite different; 
In the present case the former has a monopole-type term 
{\em \`a la} Balachandran {\em etal.} (hereafter called BMS${}^2$) \cite{Bal}, 
whose strength 
or charge is proportional to the expectation value of the quantum number 
$m$ in the state of a FV, 
$\ket{\Psi_0}$, having $(2s + 1)$-components; 
And besides the topological term contains additional terms 
that reflect the effect of interweaving components 
of a FV with their next ones. 
This is also the case for the latter; 
Such interweaving components of a FV appear in the dynamical term as well. 
In the previous version we merely showed the above results 
from the formal CSPI 
 \cite{spinPIGP}. 
It can be, however, established from the discrete CSPI. 
In the third theorem we prove that the 
general spin CS and CSPI 
contract to the general CCS and CCSPI 
in \cite{NECSPI}. 
\par
The plan of the paper is the following. 
Before going into the spin CS case, we concisely review the CCS and PI
 evolving from an arbitrary FV \cite{NECSPI} in $\S$ \ref{sec:CCSPI}. 
Next, we describe the spin CS based on arbitrary FV 
as well as their various properties ($\S$ \ref{sec:CS}) and 
employ them to perform path integration ($\S$ \ref{sec:PI}). 
Specifically, there we investigate the process of going 
from the discrete PI to the continuous ones. 
Next, we discuss problems related to the Lagrangians: 
the problems of topological terms, the fictitious gauge potentials and 
semiclassical equations. 
Complex variable form of the CS and CSPI are obtained 
in $\S$ \ref{sec:c-parametCS}. 
The results are applied to the demonstration that 
the spin CSPI there contract to the CCSPI in 
$\S$ \ref{sec:CCSPI}. 
Section \ref{sec:summary} gives the summary and prospects. 
Mathematical tools necessary in the article are enumerated concisely 
in Appendix A.
\section{General Canonical Coherent States 
and the Path Integrals}
\label{sec:CCSPI}
In this section we briefly revisit the results of the CCS with a generic FV 
and the related PI 
described in \cite{NECSPI}. 
We may compare the expressions in this section with those in the
following $\S$ \ref{sec:CS} - $\S$ \ref{sec:c-parametCS}. 
Notice that $\alpha$, in this article, denotes a parameter specifying the CCS; 
not an element of Euler angles. 
\subsection{General CCS}
\label{sec:gCCS}
We proceed physically as far as possible. 
\subsubsection{Definition of the CCS}
First, we {\em set the generic CCS, $\ket{\alpha}$, 
by displacing, or driving, 
not a usual vacuum, i.e. the zero photon
state, but an arbitrary superposition of photon number states}: 
\begin{equation}
\ket{\alpha} = {\hat D} (\alpha) \ket{\Psi_0},
\label{eqn:CCS1}
\end{equation}
where
\begin{equation}
{\hat D}(\alpha) \equiv \exp(\alpha {\hat a}^+ - \alpha^* {\hat a}) 
= \exp\Bigl( - (1 / 2) \absq{\alpha} \Bigr)
\exp(\alpha {\hat a}^+) 
\exp(- \alpha^* {\hat a})
\label{eqn:CCS-D}
\end{equation}
and 
\begin{equation}
\ket{\Psi_0} = \sum_{n=0}^{\infty} c_{n} \ket{n} 
\qquad \text{with} \qquad \sum_{n=0}^{\infty} 
\absq{c_n} = 1. 
\label{eqn:CCS-FV}
\end{equation}
Here $\ket{n}$ is the photon number state.
From (\ref{eqn:CCS1})-(\ref{eqn:CCS-FV}), 
the general CCS, $\ket{\alpha}$, can be put into the form:
\begin{equation}
\ket{\alpha} \equiv  \sum_n^\infty c_n 
\ket{\alpha, n}
\label{eqn:CCS2}
\end{equation}
with
\begin{eqnarray}
\ket{\alpha, n} & \equiv & 
{\hat D}(\alpha) \ket{n} 
 = \sum_{m=0}^{\infty} 
\bracket{m}{{\hat D} (\alpha)}
{n} \ket{m}
\nonumber\\
& = & \exp \Bigl( - (1 / 2) \absq{\alpha} 
\Bigr) \ 
\Bigl[\ \sum_{m=0}^{n}
\ \left( \frac{m!}{n!}  \right)^{1/2}
\ (- \alpha^{*})^{n-m}
\ L_{m}^{(n-m)}(\absq{\alpha}) 
\ket{m}
\nonumber\\
& \qquad  & 
+ \sum_{m = n+1}^{\infty}
\left( \frac{n!}{m!} \right)^{1/2}
\ \alpha^{m-n}\  
\ L_{n}^{(m-n)}(\vert \alpha \vert^2) \ket{m} 
\ \Bigr], 
\label{eqn:DNS}
\end{eqnarray}
 where $L_k^{(l)}(x)$ is the Laguerre polynomials. 
\par
We see that $\ket{\alpha, n}$ in (\ref{eqn:DNS}) 
is a DNS.
\footnote
{We made no mention of DNS in \cite{NECSPI}; 
Since our central concern was CSPI, we were not aware of it. 
We appreciate those who contribute to DNS 
including \cite{DNS} and \cite{AAM-C}.
} So we may say that {\em CCS with a general FV 
is an arbitrary superposition of DNS}. 
However, we will not use the explicit form of DNS in the present paper.
\subsubsection{Resolution of unity}
For CCS $\ket{\alpha}$ evolving from an 
{\em arbitrary FV}, we have the ``overcompleteness relation'' 
or ``resolution of unity'': 
\begin{equation}
\frac{1}{\pi} \int 
\ket{\alpha} d^2 \alpha \bra{\alpha} = {\bf 1} 
\quad 
{\rm with} 
\quad
d^2 \alpha \equiv d({\rm Re} \,  \alpha)\, d({\rm Im} \, \alpha).
\label{eqn:CCSresolution} 
\end{equation}
\subsubsection{CCS as eigenvectors}
It turns out that the CCS $\ket{\alpha}$ is a 
``generalized eigenvector'' \cite{Sat}
of an annihilation operator: 
\begin{equation}
({\hat a} - \alpha)^{N + 1} \ket{\alpha} = 0 
\qquad (N = \max n).
\label{eqn:CCSev1}
\end{equation}
In the case of DNS $\ket{\alpha, n}$ we also have:
\begin{equation}
(\hat a^+ -  {\alpha}^* )(\hat a - \alpha) \left\vert \alpha, n \right\rangle
= n \left\vert \alpha, n \right\rangle.
\label{eqn:CCSev2}
\end{equation}
Concerning (\ref{eqn:CCSev2}), we apologize for sign errors 
in the original expressions in Eqs. (20) and (B1) in \cite{NECSPI}. 
\subsection{General CCS path integrals}
\label{sec:genCCSPI}
\subsubsection{CCS path integrals}
\label{sec:CCSPI-form}
Invoking the resolution of unity for CCS 
(\ref{eqn:CCSresolution}), we can obtain 
PI expression for the CCS. 
We put the results below.
\par
Let us define
\begin{equation}
A({\dot \alpha}, {\dot \alpha}^{*}; \{ c_{n} \})
\equiv 2 \sum_{n=1}^{\infty} 
n^{1/2} \Bigl( 
{\dot \alpha} \, c_{n}^{*} c_{n-1} 
- {\dot \alpha}^{*} \, 
c_{n} c_{n-1}^{*} \Bigr)
\label{eqn:defA-CCS}
\end{equation}
and
\begin{equation}
H(\alpha^{*}, \alpha, t) 
\equiv \bra{\alpha}{\hat H}\ket{\alpha}.
\label{eqn:defH-CCS} 
\end{equation}
Then we find the propagator 
$K(\alpha_f, t_f; \alpha_i, t_i)$
 which starts from  $\ket{\alpha_i }$ at 
$t = t_i$ , evolves under the effect of the Hamiltonian 
$\hat H(\hat a^{+}, \hat a; t)$ which is assumed to be a suitably-ordered 
function of $\hat a^{+}$ and $\hat a$, and ends up with $\ket{\alpha_f}$ 
at $t = t_f$ is:
\begin{equation} 
K(\alpha_f, t_f; \alpha_i, t_i)
=\int {\cal D} [\alpha (t)] \  
\exp \{ (i / \hbar) S[\alpha(t)] \}
\label{eqn:propagator-CCSPI}, 
\end{equation}
where  
\begin{equation}
S[\alpha(t)] \equiv \int_{t_i}^{t_f} L dt 
\label{eqn:action-CCSPI}
\end{equation}
with 
\begin{equation}
L  \equiv 
 \frac{i \hbar}{2} 
\Bigl[ ( \alpha^{*} {\dot \alpha} - {\dot \alpha}^{*} \alpha)
+ A({\dot \alpha}, {\dot \alpha}^{*}; \{ c_{n} \}) \Bigr] 
- H(\alpha^{*}, \alpha, t) 
\label{eqn:LagCCSPI}
\end{equation}
and we symbolized
\begin{equation} 
{\cal D}[\alpha(t)] 
\equiv
\lim_{N \rightarrow \infty} 
\left( \frac{1}{\pi} \right)^N 
\prod_{j=1}^{N}  d^2 \alpha_j . 
\label{eqn:paths-CCSPI}
\end{equation}
As we see in (\ref{eqn:defA-CCS}), the extra $A$-term 
in the Lagrangian (\ref{eqn:LagCCSPI}) represents 
the entanglements of the coefficients of FV 
with their next ones. 
However, even if the $A$-term vanishes, since we take a general FV, 
$H(\alpha^{*}, \alpha, t)$ in (\ref{eqn:defH-CCS}) 
is also different from that
for the usual CCS as we showed in the evaluation of 
geometric phases for DNS in \cite{NECSPI}.
\subsubsection{Canonical equations}
\label{sec:CCSCE}
In the semiclassical limit, i.e.  $\hbar \rightarrow 0$, 
the Lagrangian (\ref{eqn:LagCCSPI}) 
yields the Euler-Lagrange equation:
\begin{equation}
i \hbar \  {\dot \alpha} 
=\ \frac{\partial H}{\partial \alpha^{*}},
\  
- i \hbar \  {\dot \alpha^{*}} 
=\ \frac{\partial H}{\partial \alpha}, 
\label{eqn:CCSvareq} 
\end{equation}
which is the generalized canonical equations. 
Since the $A$-term (\ref{eqn:defA-CCS}) is expressed as a total derivative, 
it is not involved in (\ref{eqn:CCSvareq}); 
And thus Eq. (\ref{eqn:CCSvareq}) is the same as that 
for the usual CCS formally. However, as mentioned above, 
the meaning of $\alpha$ and $H(\alpha^{*}, \alpha, t)$ 
are different. 
%
\section{SU(2) Coherent State with General Fiducial Vectors}
\label{sec:CS}
In this section we investigate the explicit form of the 
SU(2) CS based on arbitrary FV. 
And their properties are studied to such extent as we need later. 
It means that we will consider the spin states analogous to the CCS with
a generic FV in $\S$ \ref{sec:gCCS}. 
The results in $\S$ \ref{sec:CS} - $\S$ \ref{sec:c-parametCS} include 
those for the conventional SU(2)CS 
\cite{Rad,Arec,Perb} and their CSPI 
 \cite{Klau,KUS}; 
The latter follow from the former 
when we put $c_{s}=1$ and $c_m=0\ (m \ne s)$, 
or $c_{-s}=1$ and $c_m=0\ (m \ne -s)$ in later expressions. 
\subsection{Construction of the general SU(2) coherent state}
The SU(2) or spin CS are constructed from the Lie algebra satisfying 
$
{\bf S} \times {\bf S} = i \, {\bf S}, 
$ 
where ${\bf S} \equiv ({\hat S}_1, {\hat S}_2, {\hat S}_3)$ 
is a matrix vector composed of the spin operators. 
The operators in $\bf S$ are also the infinitesimal operators 
of the irreducible representation $R^{(s)}(g)$ of $SO(3)$. 
Since $SU(2) \simeq SO(3)$ locally, 
we can also use $SO(3)$ to construct the SU(2) CS.
Somewhat similar to the CCS, the SU(2)CS is defined 
by operating a rotation operator 
${\hat R}({\bf\Omega})$ with Euler angles 
${\bf\Omega} \equiv (\phi, \theta, \psi)$,
\footnote
{
Hereafter we adapt the abbreviation ${\bf\Omega} \equiv (\phi, \theta, \psi)$ 
from Radcliffe \cite{Rad} to describe a set of Euler angles 
which specifies the spin CS.
} 
which is the operator of $R^{(s)}(g)$, 
on a fixed vector $\ket{\Psi_0}$ 
in the Hilbert space of $R^{(s)}(g)$ 
\cite{Rad,Pera,Arec}:
%
\begin{equation}
\ket{\bf\Omega} \equiv \ket{\phi, \theta, \psi} 
= {\hat R}({\bf\Omega}) \ket{\Psi_0}
= \exp (- i \phi {\hat S_3}) \exp (- i \theta {\hat S_2})
 \exp (- i \psi{\hat S_3}) \ket{\Psi_0}.  
\label{eqn:CS1}                         
\end{equation}
The vector $\ket{\Psi_0}$, called a FV, 
is taken as $\ket{s, -s}$ or $\ket{s, s}$ in the conventional choice \cite{Rad,Arec}.
 
CS with such FV are closest to the classical states and have various 
useful properties 
\cite{Perb}. We appreciate them truly. 
According to the general theory of the CS, however, 
we have much wider possibilities in choosing a FV; 
And in fact it permits any normalized fixed vector in the Hilbert space 
\cite{Pera,KlaSk,Perb}. 
Thus we can take $\ket{\Psi_0}$ as
\begin{equation}
\ket{\Psi_0} = \sum_{m = - s}^{s} c_m \ket{m} 
\qquad {\rm with} \qquad
\sum_{m = - s}^{s} \absq{c_m} = 1.
\label{eqn:statvec}
\end{equation}
Hereafter $\ket{m}$ stands for $\ket{s, m}$. 
The FV will bring us all the information in later sections 
as far as the general theory goes. 
Looking at the problem in the light of physical applications, 
we need to take an appropriate $\ket{\Psi_0}$, i.e. 
$\set{c_m}$, for each system being considered. 
{\em We may consider CS with such FV as quantum states 
which have no classical analogues.} 
Exploring them surely enrich the understanding of the physical world. 
Eqs. (\ref{eqn:CS1})-(\ref{eqn:statvec}) 
correspond to (\ref{eqn:CCS1})-(\ref{eqn:CCS-FV}) 
for CCS. 
From Appendix A (ii) one can see $\ket{\bf\Omega}$ is defined on a
3-sphere $S^3$; It is specified by three real parameters, 
for which we take $\bf\Omega = (\phi, \theta, \psi)$.
Notice that the reduction of the number of Euler angles 
is not always possible for an arbitrary $\ket{\Psi_0}$; 
For any $s$, $\ket{\Psi_0}$ is not necessarily reached 
from $\ket{m}$ via ${\hat R}({\bf\Omega})$. 
Hence we use a full set of three Euler angles and proceed 
with it in what follows, which seems suitable for later discussions. 
When $\ket{\Psi_0} = \ket{m}$, we can eliminate $\psi$ from 
$\bf\Omega$, thus yielding the spin CS with the phase space 
of a 2-sphere $S^2$, the Bloch sphere, labeled by two real parameters 
$(\theta, \phi)$. 
\par
Having written $\ket{\Psi_0}$ in the form of (\ref{eqn:statvec}), 
the SU(2)CS is represented by a linear combination 
of a set of the vectors $\set{\ket{m}}$ as
\begin{equation}
\ket{\bf\Omega} 
= \sum_{m=-s}^{s} c_m \ket{{\bf\Omega}, m}
\label{eqn:CS3a}
\end{equation}
with
\begin{eqnarray} 
 \ket{{\bf\Omega}, m} 
& \equiv  & {\hat R}({\bf\Omega}) \ket{m} 
 = \sum_{m'=-s}^{s}\,  
{\sf R}_{m'{}m}^{(s)} 
({\bf\Omega})
\ \ket{m'} 
\nonumber \\
& = & \sum_{m'=-s}^{s}\, 
\exp[- i (m' \phi + m \psi)] \ {\sf r}_{m'{}m}^{(s)} 
(\theta)\ \ket{m'}. 
\label{eqn:CS3b}
\end{eqnarray}
See Appendix A (i) for the definitions of 
${\sf R}_{m'{}m}^{(s)}$ and ${\sf r}_{m'{}m}^{(s)}$. 
The form of (\ref{eqn:CS3a})-(\ref{eqn:CS3b}) is 
valuable for later arguments. 
One may see that Eqs. (\ref{eqn:CS3a})-(\ref{eqn:CS3b}) are 
analogues of Eqs. (\ref{eqn:CCS2})-(\ref{eqn:DNS}) for CCS; 
The ket $\ket{{\bf\Omega}, m}$, which corresponds to DNS $\ket{\alpha, n}$, 
may be called the ``rotated spin number state''.
\footnote
{We may, instead, call it the ``rotated magnetic quantum number state'' 
borrowing the term from spectroscopy 
\cite{Mess9}. 
However, we feel that the ``rotated spin number state'' 
sounds like a generic term and appropriate for a wide variety of spin systems.}
 We once treated $\ket{{\bf\Omega}, m}$ and its CSPI in
 \cite{PIspin}. 
\par
The state $\ket{\bf\Omega}$ may be named the ``extended spin CS'', 
yet we will call it just ``the CS'' or the ``general spin CS'' 
in this paper since there have been some arguments 
about the choice of such a FV \cite{KlaSk,Perb} 
and the CSPI 
\cite{ToKlau}.
\footnote
{
It is reviewed in \cite{Ska}. 
The authors of \cite{ToKlau} constructed ``universal propagators'' 
for various Lie group cases, being independent of the representations, 
which yields a different action from ours.
} 
We take a simple strategy for the SU(2)CSPI 
evolving from arbitrary FV here and we will give their explicit forms. 
\subsection{Resolution of unity}
\label{sec:resol}
The most important property that the CS enjoy is the 
``overcompleteness relation'' or ``resolution of unity'' 
which plays a central role in performing the path integration. 
We have the relation (\ref{eqn:CCSresolution}) for the CCS. 
In the present spin CS case, it is expressed as 
follows. 
\begin{theorem}
For $\ket{\bf\Omega}$ with arbitrary FV, we have the resolution of unity:
\begin{equation}
\int \ket{\bf\Omega} d \mu({\bf\Omega}) 
\bra{\bf\Omega} = {\bf 1}
\label{eqn:polresolution} 
\end{equation}
with
\begin{equation} 
d \mu({\bf\Omega}) 
\equiv 
\frac{2s + 1}{8\pi^2} \, d {\bf\Omega} 
\qquad 
{\rm and}
\qquad 
d {\bf\Omega}
 \equiv \sin\theta \, d\theta d\phi d\psi.
\label{eqn:polresolution2}
\end{equation}
\label{th:resUnit}
\end{theorem}
\noindent
For simplicity, we have neglected the difference between 
an integer $s$ and a half-integer $s$, which is not essential. 
Concerning the proof, 
there is an abstract way making full use of Schur's lemma
\cite{Pera,KlaSk,Perb}. 
However, we propose proving it by a slightly concrete method 
which is a natural extension of that used for the original spin CS 
\cite{Rad,Arec}. 
For it indicates clearly what is to be changed when we use a general 
$\ket{\Psi_0}$. 
\begin{proof}
We see from (\ref{eqn:CS3a})-(\ref{eqn:CS3b}) 
$
\bra{\bf\Omega}
= \sum_{{\tilde m} = -s}^{s}
\sum_{m''=-s}^{s}\, 
c_{\tilde m}^{*}  
\left( {\sf R}_{m''{}{\tilde m}}^{(s)}({\bf\Omega}) \right)^{*} 
\ \bra{m''}.
$\\
Then, with the aid of (\ref{eqn:ortho}) we have 
\begin{eqnarray}
\int \ket{\bf\Omega} d {\bf\Omega} \bra{\bf\Omega} 
&  = & \sum_{m=-s}^{s} \sum_{{\tilde m}=-s}^{s} 
c_m c_{\tilde m}^{*}  \Bigl\{ \sum_{m'=-s}^{s} 
\sum_{m''=-s}^{s}
\Bigl[ \int_0^{\pi} d \theta \ \sin\theta 
\nonumber \\ 
& & \times \int_0^{2\pi} d \phi \int_0^{2\pi} d \psi 
\left( {\sf R}_{m''{} {\tilde m}}^{(s)} ({\bf\Omega}) \right)^{*} \, 
{\sf R}_{m'{}m}^{(s)}({\bf\Omega}) 
\Bigr] \ket{m'}\bra{m''} \Bigr\} 
\nonumber \\
 & = & \sum_{m=-s}^{s} \sum_{{\tilde m}=-s}^{s}
c_m c_{\tilde m}^{*} 
\Bigl( 
\sum_{m'=-s}^{s} \sum_{m''=-s}^{s} 
\frac{8\pi^2}{2s + 1}\ \delta_{m'', m'} \delta_{{\tilde m}, m} 
\ket{m'}\bra{m''} 
\Bigr) 
\nonumber\\
&  = & \frac{8\pi^2}{2s + 1} \, 
\sum_{m=-s}^{s} 
c_m c_m^{*} 
\Bigl( \sum_{m'=-s}^{s} \proj{m'} \Bigr)
\nonumber\\
& = & \frac{8\pi^2}{2s + 1} 
( \sum_{m=-s}^{s} 
\absq{c_m} )\ {\bf 1}
= \frac{8\pi^2}{2s + 1} \ {\bf 1},
\label{eqn:proofPolresol}
\end{eqnarray}
which is exactly what we wanted. 
\end{proof}
%
\subsection{Overlap of two coherent states}
\label{sec:overlap}
The overlap of two CS 
$
\ket{{\bf\Omega}_\ell} 
\equiv \ket{\phi_\ell, \theta_\ell,  \psi_\ell} 
= \sum_{m_\ell = -s}^{s} c_{m_\ell} 
\ket{\phi_\ell, \theta_\ell, \psi_\ell; m_\ell} \  (\ell = 1, 2)
$
is one of those important quantities which we employ 
for various calculations in the CS. It can be derived, with the help of
(\ref{eqn:CS1}), (\ref{eqn:mat-R}) and (\ref{eqn:rotinverse}), as 
\begin{eqnarray}
\braket{{\bf\Omega}_2}{{\bf\Omega}_1} 
 & = & \sum_{m_1=-s}^{s} \sum_{m_2=-s}^{s}
c_{m_1} c_{m_2}^{*} 
\bra{m_2} 
{\hat R}(-\psi_2, -\theta_2, -\phi_2)  
{\hat R}(\phi_1, \theta_1, \psi_1)
\ket{m_1} 
\nonumber \\
&  = & \sum_{m_1=-s}^{s} \sum_{m_2=-s}^{s} 
c_{m_1} c_{m_2}^{*} 
{\sf R}_{m_2{}m_1}^{(s)}({\bf \Omega}_3) 
\nonumber\\
& = & \sum_{m_1=-s}^{s} \sum_{m_2=-s}^{s} 
c_{m_1} c_{m_2}^{*} 
\exp [- i (m_2 \varphi + m_1 \chi )]\ 
{\sf r}_{m_2{}m_1}^{(s)}(\vartheta),
\label{eqn:overlap1}   
\end{eqnarray}
where (\ref{eqn:mat-r}) gives the form of 
${\sf r}_{m_2{}m_1}^{(s)}(\vartheta)$ 
and ${\bf \Omega}_3 \equiv (\varphi, \vartheta, \chi)$ 
 is determined by (\ref{eqn:tworots}) 
 if we replace $\tilde {\bf \Omega}$ with ${\bf \Omega}_3$. 
 It is easy to see that any state 
$\ket{\bf\Omega}$ is normalized to unity, 
as conforms to our construction of the CS. 
\subsection{Typical matrix elements}
Typical matrix elements that we may employ in later are: 
\begin{equation}
\left\{
\begin{array}{l}
\bra{\bf\Omega} 
{\hat S}_{3} 
\ket{\bf\Omega} 
= A_0(\set{c_m}) \cos\theta    
-  A_1(\psi; \set{c_m}) \sin\theta 
 \\
\bra{\bf\Omega} 
{\hat S}_{+} 
\ket{\bf\Omega} 
= A_0(\set{c_m}) \sin\theta  \exp(i \phi)  
+ A_2({\bf\Omega}; \set{c_m})  
= \bra{\bf\Omega} 
{\hat S}_{-} 
\ket{\bf\Omega}^{*},
\end{array}
\right.
\label{eqn:mat-pol}
\end{equation}
where
$
{\hat S}_{\pm} = {\hat S}_1 \pm i {\hat S}_2 
$
and 
\begin{equation}
\left\{
\begin{array}{l}
A_0(\set{c_m}) 
=\sum_{m=-s}^{s} m \absq{c_m}  
\\
A_1(\psi; \set{c_m})  
= \frac12 \sum_{m=-s+1}^{s} f(s, m) 
[ c_m^{*} c_{m-1} \exp(i \psi) + c_m c_{m-1}^{*} 
\exp(- i \psi) ]
\\
A_2({\bf\Omega}; \set{c_m})
= \frac12 \sum_{m=-s+1}^{s} f(s, m) 
\exp(i \phi) 
\{ (1 + \cos\theta) \exp(i \psi) c_m^{*} c_{m-1} 
\\ 
\qquad \qquad - (1 - \cos\theta) \exp(- i \psi) 
c_m c_{m-1}^{*} \}
\\
f(s, m) 
= [(s+m)(s-m+1)]^{1/2}. 
\end{array}
\right.
\label{eqn:defA}
\end{equation}
By $\set{c_m}$ we mean a set of the coefficients of a FV. 
We can easily verify (\ref{eqn:mat-pol}) 
by (\ref{eqn:CS1}) and (\ref{eqn:RS}). 
\par
Generating functions for general matrix elements exist 
as in the original CS cases 
\cite{Arec}. 
In the normal product form it reads 
\begin{eqnarray}
X_N (z_+, z_3, z_-) 
& \equiv & 
\bra{{\bf\Omega}_2} \, 
\exp(z_+ {\hat S_+}) \, \exp(z_3 {\hat S}_3) \, 
\exp(z_- {\hat S_-}) \, 
\ket{{\bf\Omega}_1} 
\nonumber \\
& = & \bra{\Psi_0} {\hat R}^{+}({\bf\Omega}_2) 
{\hat R}({\bf\Omega}) {\hat R}({\bf\Omega}_1) \ket{\Psi_0} 
= \bra{\Psi_0} {\hat R}({\bf\Omega}^{''}) 
\ket{\Psi_0}, 
\label{eqn:genFunc}
\end {eqnarray}
where ${\bf\Omega}$ is related to $z_\ell \ (\ell = +, 3, -)$ 
through (\ref{eqn:EulerComp}); 
And ${\bf\Omega}^{''}$ is determined by (\ref{eqn:trirots}) 
with suitable changes. 
In principle any matrix elements can be obtained from (\ref{eqn:genFunc})
via partial differentiations with respect to appropriate 
variables $z_i\, (i = \pm, 3)$. 
\section{Path Integral via the General Spin CS}
\label{sec:PI}
We now give the PI expressions in terms of the general spin CS evolving 
from an arbitrary FV defined in $\S$ \ref{sec:CS}. 
In $\S$ \ref{sec:PI1} the PI form is given. 
It is proved in the following $\S$ \ref{sec:d2cPI}. 
Some specific aspects of the Lagrangian are discussed 
in $\S$ \ref{sec:topol-term} - $\S$ \ref{sec:gauge-pot}. 
\subsection{Path integrals}
\label{sec:PI1}
%
In this section we will give the explicit PI expression of the 
transition amplitude by means of the CS discussed in $\S$ \ref{sec:CS}. 
What we need is the propagator 
$K({\bf\Omega}_{f}, t_{f}; {\bf\Omega}_{i}, t_{i})$
 which starts from  $\ket{{\bf\Omega}_{i}}$ at $t=t_{i}$, evolves 
 under the effect of the Hamiltonian 
${\hat H} ({\hat S_{+}}, {\hat S_{-}}, {\hat S_{3}}; t)$ 
which is assumed to be a function of $\hat S_{+}$, 
$\hat S_{-}$ and $\hat S_{3}$ with a suitable operator ordering 
and ends up with $\ket{{\bf\Omega}_{f}}$ at $t=t_{f}$: 
\begin{equation}
K({\bf\Omega}_f, t_f; {\bf\Omega}_i, t_i)
=\braket{{\bf\Omega}_f, t_f}{{\bf\Omega}_i, t_i}
=\bra{{\bf\Omega}_f} \, 
{\rm T} \exp [- (i / \hbar) 
\int_{t_i}^{t_f} {\hat H}(t) d t ] \, 
\ket{{\bf\Omega}_i},
\label{eqn:propagator1}
\end{equation}
where T denotes the time-ordered product. 
The overcompleteness relation (\ref{eqn:polresolution}) affords us 
the well-known prescription of formal CSPI
 \cite{KlaSk,IKG} to give 
\begin{equation} 
K({\bf\Omega}_f, t_{f}; {\bf\Omega}_i, t_{i})
=\int  
\exp \{ (i / \hbar) S[{\bf\Omega}(t)] \} \, 
{\cal D} [{\bf\Omega}(t)],
\label{eqn:PI}
\end{equation}
where
\begin{equation}
S[{\bf\Omega}(t)]
\equiv \int_{t_i}^{t_f}   
\Bigl[ \ \bra{\bf\Omega} i \hbar 
\frac{\partial}{\partial t} 
 \ket{\bf\Omega} 
 - H({\bf\Omega}, t) \ \Bigr] 
\, d t 
 \equiv \int_{t_i}^{t_f}   
 L({\bf\Omega}, {\bf \dot \Omega},t) \, d t 
\label{eqn:action-pol}
\end{equation}
with 
\begin{equation}
H({\bf\Omega}, t) 
\equiv 
\bracket{{\bf\Omega}}{\hat H}{{\bf\Omega}}
\label{eqn:H}
\end{equation}
and we symbolized
\begin{equation}
 {\cal D}[{\bf\Omega}(t)] 
\equiv
\lim_{N \rightarrow \infty} 
\prod_{j=1}^{N} d \mu({\bf\Omega}_{t_j}) 
\equiv
\prod_t \frac{8\pi^2}{2s + 1} 
\, [\sin\theta(t) d \theta(t)
d \phi(t) d \psi(t)].
\label{eqn:paths-pol}
\end{equation} 
The explicit form of the Lagrangian yields
\begin{equation}
L({\bf\Omega}, {\dot {\bf\Omega}}, t) 
= \hbar \Bigl[ A_0(\set{c_m}) ({\dot \phi}\cos\theta + {\dot \psi}) 
+ A_3({\bf\Omega}, {\dot {\bf\Omega}}; \set{c_m}) \Bigr] 
- H({\bf\Omega}, t),
\label{eqn:Lagpol}
\end{equation}
where 
\begin{equation}
A_3({\bf\Omega}, {\dot {\bf\Omega}}; \set{c_m}) 
\equiv 
- A_1(\psi; \set{c_m}) \, {\dot \phi} \sin\theta 
+ A_4(\psi; \set{c_m}) \, {\dot \theta}
\label{eqn:deffA3}
\end{equation}
with $A_1(\psi; \set{c_m})$ in (\ref{eqn:defA}) and 
\begin{equation}
A_4(\psi; \set{c_m}) 
\equiv 
\frac{1}{2 i} 
\sum_{m=-s+1}^{s} f(s, m) 
[ c_m^{*} c_{m-1} \exp(i \psi)  
- c_m c_{m-1}^{*} \exp(- i \psi) ]. 
\label{eqn:def-A4}
\end{equation}
\par
The formal proof of (\ref{eqn:Lagpol}) is best carried 
out by the use of the identity:
\begin{eqnarray}
{\hat R}^{+}({\bf\Omega}) 
\frac{\partial}{\partial t} 
{\hat R}({\bf\Omega}) 
& = & {\hat R}^{+}({\bf\Omega})  
\Bigl( \, {\dot \phi} \frac{\partial}{\partial \phi}
+ {\dot \theta} \frac{\partial}{\partial \theta} 
+ {\dot \psi} \frac{\partial}{\partial \psi} 
\, \Bigr) {\hat R}(\bf\Omega) 
\nonumber \\
& =  & - i ({\dot \phi} \cos\theta + {\dot \psi}){\hat S_3}
+ \frac12 ( i {\dot \phi} \sin\theta - {\dot \theta} ) 
\exp(i \psi){\hat S_+} 
\nonumber \\
& & \qquad 
+ \frac12 ( i {\dot \phi} \sin\theta + {\dot \theta} ) 
\exp(- i \psi){\hat S_-}
\label{eqn:pol-R}
\end{eqnarray}
and (\ref{eqn:mat-pol}). 
Since the relation (\ref{eqn:pol-R}) is independent of $s$, 
it can be readily verified by the use of a $2 \times 2$ matrix (\ref{eqn:s1/2}). 
The detailed and substantial proof of (\ref{eqn:Lagpol}) 
is given in the following subsection $\S$ \ref{sec:d2cPI}. 
\subsection{From discrete to continuous path integrals}
\label{sec:d2cPI}
We can justify the spin CSPI in $\S$ \ref{sec:PI1} by showing 
the process from the discrete PI to the continuous ones. 
The proof is a straightforward generalization of 
that in \cite{PIspin}. 
The method below may be applied to other CSPI; 
For example, the SU(1, 1)  or SU(3)  cases. 
\begin{theorem}
The quantum evolution of a physical system in terms of the general SU(2) CS 
is represented by (\ref{eqn:PI})-(\ref{eqn:def-A4}).
\end{theorem}
\begin{proof}
By dividing the time interval into infinite numbers of 
an infinitesimal one $\epsilon$ in (\ref{eqn:propagator1}) 
and the successive use of the overcompleteness relation, 
i.e. Eq. (\ref{eqn:polresolution}) in Theorem \ref{th:resUnit}, 
we obtain
\begin{eqnarray}
K({\bf \Omega}_{f}, t_{f}; {\bf \Omega}_{i}, t_{i}) 
& = & \bra{{\bf\Omega}_f} \, 
{\rm T} \exp [- (i / \hbar) \int_{t_i}^{t_f} {\hat H}(t) \, d t ] \, 
\ket{{\bf\Omega}_i} 
\nonumber \\
& =& \lim_{N \rightarrow \infty} \int d \mu ({{\bf \Omega}_{1}}) \cdots 
\int d \mu ({{\bf \Omega}_{N}}) 
\braket{{\bf \Omega}_{f}, t_{f}}{{\bf \Omega}_N, t_N} \cdots
\nonumber\\
& & \qquad \qquad 
\times \braket{{\bf \Omega}_{j}, t_{j}}{{\bf \Omega}_{j -1}, t_{j -1}}
\cdots
\braket{{\bf \Omega}_{1}, t_{1}}{{\bf \Omega}_{i}, t_{i}},
\label{eqn:propagator2} 
\end{eqnarray}
where 
$
\epsilon = [1 / (N+1)] (t_{f} - t_{i})
$ 
and 
$
t_{j} = t_{i} + j \epsilon.
$ 
It is clear that we only have to consider a propagator 
during an infinitesimal time interval, which gives 
\begin{eqnarray}
\braket{{\bf \Omega}_j, t_j}{{\bf \Omega}_{j -1}, t_{j -1}} 
& \equiv & \bra{{\bf\Omega}_j} \, 
 {\rm T} \exp[- (i / \hbar) \int_{t_{j -1}}^{t_j} {\hat H}(t) \, d t ]\,
\ket{{\bf \Omega}_{j-1}}
\nonumber \\ 
& \simeq &
\bra{{\bf \Omega}_j} 
\Bigl( \, 1- (i / \hbar) \int_{t_{j-1}}^{t_j} d t 
\ {\hat H}({\hat S}_{+}, {\hat S}_{-}, {\hat S}_3;t) 
\, \Bigr)  
\ket{{\bf \Omega}_{j-1}}
\nonumber \\
&= &\braket{{\bf \Omega}_j}{{\bf \Omega}_{j -1}}
\Bigl(\, 1- (i / \hbar) \, \epsilon \, 
H({\bf \Omega}_{j}, 
{\bf \Omega}_{j-1}; t_{j-1}) \, \Bigr) 
\nonumber\\
& \simeq & 
\exp[ \ln \braket{{\bf \Omega}_j}{{\bf \Omega}_{j -1}}] 
\exp\Bigl[\, - (i / \hbar) \, \epsilon \, 
H({\bf \Omega}_{j}, 
{\bf \Omega}_{j-1}; t_{j-1}) \, \Bigr], 
\label{eqn:itprop}
\end{eqnarray}
where 
\begin{equation}
H({\bf\Omega}^{''}, {\bf\Omega}^{'}; t)
\equiv
\frac{
\bra{\bf\Omega^{''}}
{\hat H}({\hat S}_{+}, {\hat S}_{-}, {\hat S}_3;t)
\ket{{\bf\Omega}^{'}}
}
{
\braket{{\bf\Omega}^{''}}{{\bf\Omega}^{'}}
}.
\label{eqn:H-Omega}
\end{equation}
\par
Our next task is to compute the infinitesimal overlap 
$\braket{{\bf \Omega}_j}{{\bf \Omega}_{j -1}}$ 
in (\ref{eqn:itprop}). 
Representing 
$\ket{{\bf \Omega}_j} 
= \sum_m c_m {\hat R}({\bf \Omega}_j) \ket{m}$ 
and 
$\ket{{\bf \Omega}_{j - 1}} 
= \sum_{m'} c_{m'} {\hat R}({\bf \Omega}_{j - 1}) \ket{m'}$,  
we see from (\ref{eqn:overlap1}) 
\begin{eqnarray}
\braket{{\bf\Omega}_j}{{\bf\Omega}_{j-1}} 
& = & \sum_{m =-s}^s \sum_{m' =-s}^s 
c_m^* c_{m'}
\bra{m} {\hat R}({\tilde {\bf\Omega}_j}) \ket{m'}
\nonumber \\
& = & \sum_{m =-s}^s \sum_{m' =-s}^s 
c_m^* c_{m'} 
\exp[-i ( m {\tilde \phi_j} + m' {\tilde \psi_j}) ] 
\, {\sf r}_{m{}m'}({\tilde \theta_j}),
\label{eqn:it-overlap} 
\end{eqnarray}
where ${\tilde {\bf\Omega}_j }
\equiv ({\tilde \phi_j}, {\tilde \theta_j}, 
{\tilde \psi_j})$ 
satisfies the same relation (\ref{eqn:overlap1}) 
as ${\bf\Omega}_3$ 
if we put ${\bf\Omega}_2 = {\bf\Omega}_j = (\phi_j, \theta_j, \psi_j)$
and ${\bf\Omega}_1 = {\bf\Omega}_{j - 1} = (\phi_{j - 1}, \theta_{j-1}, \psi_{j-1})$. 
Searching for concrete relations between 
${\tilde {\bf\Omega}_j }$, ${\bf\Omega}_j $ and 
${\bf\Omega}_{j - 1}$, we are brought to (\ref{eqn:tworots}). 
Then, by the use of the relation:
$                                                                       
\Delta \theta_{j} 
\equiv \theta_j - \theta_{j-1} 
\simeq {\dot \theta}_{j} \epsilon, 
\Delta \phi_{j} \equiv \phi_j - \phi_{j-1} 
\simeq {\dot \phi_j} \epsilon 
$
and 
$
\Delta \psi_{j} 
\equiv \psi_j - \psi_{j-1} \simeq {\dot \psi}_{j} \epsilon
$ 
we have, to ${\rm O}(\epsilon)$, 
\begin{equation}
\left\{
\begin{array}{l}
\tilde \theta_j \simeq \Delta \theta_j, 
\quad 
\cos(\frac12 {\tilde \theta_j})  \simeq 1,
\quad 
\sin(\frac12 {\tilde \theta_j})  \simeq 
\frac12 \sin({\tilde \theta_j}) 
\simeq  {\rm O}(\epsilon), 
 \\
\sin{\tilde \theta_j} \exp( i {\tilde \phi_j} ) 
\simeq  - \exp( - i \psi_j ) 
( \Delta\theta_j + i \Delta\phi_j \sin\theta_j ),
\\
\exp[ i ({\tilde \phi_j} + {\tilde \psi_j}) ] 
 \simeq 1 - i (\Delta\phi_j \cos\theta_j 
+ \Delta\psi_j).
\end{array}
\right.
\label{eqn:it-tworots}
\end{equation}
The relation (\ref{eqn:it-tworots}) above is the master key 
to the proof which we invoke implicitly 
in (\ref{eqn:it-overlap1})-(\ref{eqn:it-overlap3}) below. 
Now we have from (\ref{eqn:mat-r}) 
\begin{eqnarray}
{\sf r}_{m{}m'}({\tilde \theta_j}) 
& = & \sum_t N(s, m, m'; t) 
\cdot [\cos({\tilde \theta_j} / 2)]
^{(s + m - t) + (s - m' - t)} 
\nonumber \\
& & \qquad \qquad 
\times [\sin({\tilde \theta_j} / 2)]^{(t) + (t - m + m')}.
\label{eqn:N2}
\end{eqnarray}
From the restrictions on the factorials 
in $N(s, m, m'; t)$ in (\ref{eqn:N(s, m, m'; t)}) 
all the terms in the parentheses in the 
exponents in (\ref{eqn:N2}) 
have to be zero or more. 
Then it becomes apparent that 
we only need to pick up, to ${\rm O}(\epsilon)$, 
the terms fulfilling the combinations 
$(t, t - m + m') = (0, 0), (0, 1), (1, 0)$, 
from which the following three cases 
appear: 
(i)$t = 0, m = m'$,
(ii)$t = 0, m' = m + 1$,
(iii)$t = 1, m' = m - 1$. 
Evaluating terms in (\ref{eqn:it-overlap}), 
the following expressions for each case 
are derived by (\ref{eqn:it-tworots}), 
(\ref{eqn:N2}) and (\ref{eqn:N(s, m, m'; t)}). 
First, consider the case (i). 
Notice that 
$N(s, m, m' = m; t = 0) = 1$ and 
${\sf r}_{m{}m}({\tilde \theta_j}) 
\simeq 1 + {\rm O} \Bigl( [\Delta(\theta_j)]^2 \Bigr)$. 
The relevant terms are 
\begin{equation}
\sum_{m=-s}^{s} \absq{c_m}  
\exp[- i \, m ({\tilde \phi_j} + {\tilde \psi_j}) ] 
\, {\sf r}_{m{}m}({\tilde \theta_j}) 
\simeq
1 + i  \, A_0(\set{c_m}) 
(\Delta\phi_j \cos\theta_j 
+ \Delta\psi_j). 
\label{eqn:it-overlap1}
\end{equation}
Second, for the case (ii) $N(s, m, m' = m + 1; t = 0) 
= f(s, m')$; See (\ref{eqn:defA}) for the definition 
of $f(s, m)$. 
So one finds that the corresponding terms become 
\begin{eqnarray}
 &  & \frac12 \sum_{m' = -s +1}^{s} c_{m'} c_{m'-1}^* 
\exp[ - i \, m' ({\tilde \phi}_j + {\tilde \psi}_j) ] 
f(s, m') \cdot [\sin{\tilde \theta}_j \exp( i \, {\tilde \phi_j} )]
\nonumber \\
& \simeq &
\frac12 \sum_{m = -s +1}^{s} c_{m} c_{m - 1}^* 
[1 + i \, m (\Delta\phi_j \cos\theta_j + \Delta\psi_j) ]
\nonumber \\
& & \qquad 
\times f(s, m) [ - \exp( - i \psi_j ) 
(\Delta\theta_j + i \, \Delta\phi_j \sin\theta_j) ]
\nonumber \\
&  \simeq &  -\frac12 \sum_{m = -s +1}^{s} c_{m} c_{m - 1}^*
f(s, m) \exp( - i \psi_j ) 
(\Delta\theta_j + i \, \Delta\phi_j \sin\theta_j) 
+ {\rm O}(\epsilon^2),
\label{eqn:it-overlap2}
\end{eqnarray}
where we have renamed $m'$ $m$. 
Third, we can deal with the case (iii) 
just in the same manner as (ii). 
With $N(s, m, m' = m - 1; t = 1) 
= - f(s, m)$  the corresponding terms become 
\begin{eqnarray}
& - & \frac12 \sum_{m = -s +1}^{s} c_m^* c_{m-1} 
f(s, m) 
\exp[- i \, (m - 1) 
({\tilde \phi}_j + {\tilde \psi}_j)] 
\cdot 
[ \sin{\tilde \theta}_j \exp(- i \, {\tilde \phi_j}) ] 
\nonumber \\
 & \simeq &
- \frac12 \sum_{m = -s +1}^{s} c_{m}^* c_{m - 1} 
 [1 + i \, (m - 1) (\Delta\phi_j \cos\theta_j 
+ \Delta\psi_j) ]
\nonumber \\
& & \qquad \qquad 
\times f(s, m) [ - \exp( i \psi_j ) 
(\Delta\theta_j - i \, \Delta\phi_j \sin\theta_j) ]
\nonumber  \\
& \simeq &  \frac12 \sum_{m = -s +1}^{s} c_{m}^* c_{m - 1} 
f(s, m) \exp(i \psi_j ) 
(\Delta\theta_j - i \, \Delta\phi_j \sin\theta_j) 
+ {\rm O}(\epsilon^2).
\label{eqn:it-overlap3}
\end{eqnarray}
Eventually, putting the above results all together, 
we obtain the infinitesimal overlap:
\begin{eqnarray}
\braket{{\bf\Omega}_j}{{\bf\Omega}_{j-1}} 
& \simeq & 
1 + i A_0(\set{c_m}) (\Delta\phi_j \cos\theta_j + \Delta\psi_j)
\nonumber \\ 
& & - \frac12 \sum_{m = -s +1}^{s} 
f(s, m) [ c_{m} c_{m - 1}^* 
\exp( - i \psi_j ) (\Delta\theta_j 
+ i \, \Delta\phi_j \sin\theta_j) 
\nonumber  \\
& &  - c_{m}^* c_{m - 1} 
\exp( i \psi_j ) (\Delta\theta_j 
- i \, \Delta\phi_j \sin\theta_j) ].
\label{eqn:it-overlap-f}
\end{eqnarray}
\par 
Substituting (\ref{eqn:it-overlap-f}) into (\ref{eqn:itprop}) 
and then (\ref{eqn:itprop}) into (\ref{eqn:propagator2}), 
we finally arrive at the expression to ${\rm O}(\epsilon)$: 
\begin{equation}
K({\bf\Omega}_{f}, t_{f}; {\bf\Omega}_{i}, t_{i}) 
= \lim_{N \rightarrow \infty} \int d \mu ({\bf\Omega}_{1}) 
\cdots \int d \mu ({\bf\Omega}_{N}) 
\exp[(i / \hbar) \, S_{1, N+1}]
\label{eqn:propagator3} 
\end{equation}
with
\begin{eqnarray}
S_{1, N+1} & = & \sum_{j = 1}^{N+1} 
\Bigl\{ 
\, 
\hbar 
\Bigl[ A_0(\set{c_m}) 
(\Delta\phi_j \cos\theta_j + \Delta\psi_j) 
\nonumber \\
& & \qquad  - A_1(\psi_j; \set{c_m}) \Delta\phi_j \sin\theta_j 
+ A_4(\psi_j; \set{c_m}) \Delta\theta_j \Bigr]
\nonumber\\
& &  
\qquad  
 - \epsilon \, H({\bf\Omega}_{j}, 
{\bf\Omega}_{j-1}; t_{j-1}) \Bigr\} .
\label{eqn:ds-action}
\end{eqnarray} 
See (\ref{eqn:defA}) and (\ref{eqn:def-A4}) 
for the definitions of $A_1$ and $A_4$ respectively. 
Hence, it is easy to see that the expressions 
(\ref{eqn:propagator3})-(\ref{eqn:ds-action}) agree with 
those of (\ref{eqn:PI})-(\ref{eqn:def-A4}) in 
the $\epsilon \rightarrow 0$ limit. \\
\end{proof}
\noindent
We have thus arrived at {\em the generic expressions of the PI 
via the SU(2)CS}, i.e. (\ref{eqn:PI})-(\ref{eqn:def-A4}), 
which constitute one of the main results of the present paper. 
They correspond to (\ref{eqn:defA-CCS})-(\ref{eqn:paths-CCSPI}) for the CCS. 
The complex variable forms of spin CS and CSPI are presented 
in later section ($\S$ \ref{sec:c-parametCS}). 
There, with the aid of the form, 
we can easily see the above correspondence 
by the contraction procedure.  
The special case when $c_m = 1$ (for a sole $m$), 
which includes the conventional SU(2)CSPI, 
was once treated in \cite{PIspin}.
\par 
The transition amplitude between any two states $\ket{i}$ 
at $t = t_{i}$ 
and $\ket{f}$ at $t = t_{f}$ can be evaluated by 
\begin{equation}
\int \!\!\! \int d \mu({\bf\Omega}_f)
d \mu({\bf\Omega}_i) \ 
\braket{f}{{\bf\Omega}_{f}} 
({\bf\Omega}_{f}, t_{f}; {\bf\Omega}_{i}, t_{i})
\braket{{\bf\Omega}_{i}}{i}.
\label{eqn:tranamp}
\end{equation}
\par
Finally, we will make an auxiliary discussion on the derivation of CSPI. 
We often see a slightly different approach to PI in literature
\cite{Ryder,SaKi,Weinberg1-PI}; 
It is essentially after an original one due to Dirac 
\cite{Dirac-Sov,Dirac-QM}. 
Now let us proceed with it. 
Define a ``moving frame'' state vector 
\cite{Ryder,SaKi}, 
which represents an intermediate state, as
\begin{equation}
\ket{{\bf\Omega}_j, t_j} 
\equiv 
{\hat U}^+ (t_j) \ket{{\bf\Omega}_j} 
\qquad 
{\rm with} 
\qquad 
{\hat U} (t_j) 
\equiv 
{\rm T} \exp[- (i / \hbar) 
\int_{t_i}^{t_j} {\hat H}(t) \, d t ].
\label{eqn:mv-frame-ket}
\end{equation}
Then the resolution of unity also holds for 
$\ket{{\bf\Omega}_j, t_j}$ 
whose successive use leads us formally 
to the same expression as (\ref{eqn:propagator2}). 
However, this time we have, 
as a {\em consequence} of (\ref{eqn:mv-frame-ket}), 
\begin{equation}
\braket{{\bf \Omega}_j, t_j}{{\bf \Omega}_{j -1}, t_{j -1}} 
= \bra{{\bf \Omega}_j} \, 
{\rm T} \exp[- (i / \hbar) 
\int_{t_{j -1}}^{t_j} {\hat H}(t) \, d t ]\,
\ket{{\bf \Omega}_{j-1}}.
\label{eqn:overlap-mv-frames}
\end{equation}
In contrast the same expression is used as a {\em definition} 
in (\ref{eqn:itprop}). 
Note that in ${\hat U}^+ (t_j)$ the order of a set of operators 
$\set{\exp[(i / \hbar) 
\int_{t_{j - 1}}^{t_j} {\hat H}(t) d t]}$ 
is reversed to that in ${\hat U} (t_j)$; 
It becomes anti-chronological. 
The residuary procedure to CSPI is the same as the former. 
\par
In the phase space PI ``moving frame'' vectors $\set{\ket{q_j, t_j}}$ 
are defined as the eigenstates of the operators 
${\hat q}(t_j) \equiv {\hat U}^+(t_j) {\hat q} {\hat U}(t_j)$ 
in the Heisenberg picture 
\cite{SaKi,Weinberg1-PI}, 
from which
$ \ket{q_j, t_j} = U^+ (t_j) \ket{q_j}$  results. 
They are truly the precise descriptions of intermediate states. 
In the present spin PI case we can do the same thing 
for a simple FV $\ket{{\bf\Psi}_0} = \ket{m}$: 
$\ket{{\bf\Omega}_j, t_j}$ may be 
defined as the eigenstate of the operator 
$[{\hat R} ({\bf\Omega}) 
{\hat S}_3 {\hat R}^{+} ({\bf\Omega})]_{t = t} 
\equiv {\hat U}^+ (t) \cdot [{\hat R}({\bf\Omega}) 
{\hat S}_3 {\hat R}^{+} ({\bf\Omega}) ]_{t = 0} 
\cdot {\hat U}(t)$ 
which is in the Heisenberg picture. It corresponds to (\ref{eqn:CCSev2})
 for the CCS. However, for a generic FV, we are not able to interpret 
$\ket{{\bf\Omega}_j, t_j}$ 
as eigenvectors of some operators no more. 
And thus we have adopted the former approach 
which seems more plausible to the generic CSPI 
in the sense. 
Of course, {\em for any FV} CS are clearly defined 
and the overcompleteness relation holds as (\ref{eqn:polresolution}); 
Indeed it is almost the only relation that CS enjoy 
\cite{KlaSk}. 
So we can perform PI as we saw it.
 The relation is the fundamental 
feature of CS that makes CS such a flexible tool 
for PI and that makes CSPI so fascinating. 
%
%
\subsection{The topological term}
\label{sec:topol-term}
The term with the square brackets in 
the Lagrangian (\ref{eqn:Lagpol}), 
\begin{equation} 
A_0(\set{c_m}) ({\dot \phi}\cos\theta + {\dot \psi}) 
+ A_3({\bf\Omega}, {\dot {\bf\Omega}}; \set{c_m}),
\label{eqn:topol-term0}
\end{equation}
stemming from $\bra{\bf\Omega}(\partial / \partial t) 
\ket{\bf\Omega}$, 
may be called the ``topological term'' 
that is related to the geometric phases.
\footnote
{
The significance of the term was once recognized by 
 Kuratsuji, who called it the ``canonical term'', 
in relation to the semiclassical quantization; 
 note that the geometric phase associated with the term was called the 
 ``canonical phase'' in \cite{KMa} and \cite{KraC}; 
See \cite{KraC}  and references therein. 
We call them just the geometric phases  in the present paper.
}
Here the $A_3$-term is given by (\ref{eqn:deffA3}). 
And one can see that the first term in the topological term 
gives a description of monopoles {\em \`a la} 
BMS${}^2$  \cite{Bal,Aitch}. 
The fictitious gauge potentials corresponding 
to the whole topological term are also non-singular as 
\cite{Bal} and \cite{Aitch}; 
See $\S$ \ref{sec:gauge-pot} for the point. 
\par
In the differential 1-form, the whole topological term $\kappa$ reads: 
\begin{equation}
\kappa  =  A_0(\set{c_m}) \, 
(\cos\theta \, d \phi + d \psi) 
- A_1(\psi; \set{c_m}) \, \sin\theta \, d \phi 
+ A_4(\psi; \set{c_m}) \, d \theta
\label{eqn:topol-term-diff1}
\end{equation}
and in the 2-form:
\begin{eqnarray}
d \kappa 
& = & - \Bigl[
A_0(\set{c_m}) \sin\theta 
+ A_1(\psi; \set{c_m}) \cos\theta 
\Bigr] \, 
d \theta \wedge d \phi
\nonumber \\ 
& & \qquad 
- A_4(\psi; \set{c_m}) \, \sin\theta \, 
d \phi \wedge d \psi 
+ A_1(\psi; \set{c_m}) \, 
d \psi \wedge d \theta.
\label{eqn:topol-term-diff2}
\end{eqnarray}
One may see that {\em the strength of the well-known 
monopole-type term depends on $A_0$}, i.e. 
the expectation value of the 
quantum number $m$ in the state of $\ket{\Psi_0}$. 
{\em In addition we have other fields 
with $A_1$ and $A_4$-terms describing the effects of 
interweaving coefficients of $\ket{\Psi_0}$ with their next ones.} 
We have thus obtained  the {\em general expressions of the topological term
in the SU(2)CSPI},
i.e. (\ref{eqn:topol-term-diff1})-(\ref{eqn:topol-term-diff2}), 
which are also one of the main results of the present paper. 
\par
For a FV with $c_m = 1$ (for a sole $m$), 
since $A_1$ and $ A_4$-terms vanish 
we have 
\begin{equation}
\kappa  =  m \, (\cos\theta \, d \phi + d \psi), 
\qquad 
d \kappa = - m \sin\theta \, d \theta \wedge d \phi,
\label{eqn:topol-term-diff3}
\end{equation}
which represents a monopole with the strength $m$. 
\subsection{Semiclassical limit}
\label{sec:sc-spinCSPI}
In this subsection  we will investigate what information 
the semiclassical limit of CSPI brings. 
In the situation where $\hbar \ll S[{\bf\Omega}(t)]$, 
the principal contribution in (\ref{eqn:PI}) 
comes from the path that satisfies $\delta S = 0$, which requires 
the Euler-Lagrange equations for 
$L({\bf\Omega}, {\bf \dot \Omega},t)$. 
Then we obtain
\begin{equation}
\left
\{
\begin{array}{l}
\hbar 
\{ 
[A_0(\set{c_m}) \sin\theta 
+ A_1(\psi; \set{c_m}) \cos\theta 
] {\dot \phi} 
+ A_1(\psi; \set{c_m}) {\dot \psi}
\} 
= - ( \partial H / \partial \theta )
\\
\hbar 
\{
[
A_0(\set{c_m}) \sin\theta 
+ A_1(\psi; \set{c_m}) \cos\theta 
] {\dot \theta} 
- [A_4(\psi; \set{c_m}) \sin\theta] {\dot \psi}  
\} 
= \partial H / \partial \phi
\\
\hbar 
\{
[A_4(\psi; \set{c_m}) \sin\theta] {\dot \phi}  
+ A_1(\psi; \set{c_m}) {\dot \theta} 
\} 
= {\partial H / \partial \psi},
\label{eqn:vareq-pol} 
\end{array}
\right.
\end{equation}
where $A_0$ and $A_1$ are given by (\ref{eqn:defA}) 
and $A_4$ is defined by (\ref{eqn:def-A4}).
\par
The expressions in (\ref{eqn:vareq-pol}) 
are the variational equations for the spin CS parameters $\bf\Omega$, 
which may be compared with (\ref{eqn:CCSvareq}) 
in CCSPI. 
\par
The special case, 
i.e. that for SU(2)CS with a FV 
$\ket{\Psi_0} = \ket{m}$, was once treated in 
\cite{PIspin}; 
Putting $\ket{\Psi_0} = \ket{-s}$ 
brings us back to the results for the original case 
 \cite{Klau,KUS}.
%
\subsection{The nature of fictitious gauge potentials}
\label{sec:gauge-pot}
We now investigate the nature of fictitious gauge potentials 
corresponding to the whole topological term $\kappa$ in 
(\ref{eqn:topol-term-diff1}). 
We follow the strategy in \cite{Aitch}. 
\par
Using the orthogonal coordinates: 
\begin{equation}
\xi = \phi + \psi,
\qquad
\eta = \phi - \psi
\label{eqn:def-ort-coords}
\end{equation}
with (\ref{eqn:real-R}), 
the metric is 
\begin{equation}
d s^2 = d {\bf x}^2 
= \frac14 [ d \theta^2 + \cos^2 (\theta / 2) d \xi^2 
+ \sin^2 (\theta / 2) d \eta^2 ] 
\equiv d {\bf n}^2, 
\label{eqn:metric}
\end{equation}
where ${\bf n}$ stands for a unit vector 
in the $(\theta, \xi, \eta)$-coordinates. 
Now let us call the fictitious gauge potential 
${\bf {\tilde A}} 
= ({\tilde A}_\theta, {\tilde A}_\xi, {\tilde A}_\eta)$; 
We use ${\bf \tilde A}$ so as not to confuse them with 
$A_i\ (i = 1, \cdots, 4)$-terms in $\S$ \ref{sec:CS} 
- $\S$ \ref{sec:PI}. 
Then we have 
\begin{equation}
\kappa = {\bf {\tilde A} \cdot} d {\bf n}
= {\tilde A}_\theta \, \frac12 d \theta 
+ {\tilde A}_\xi \, \frac12 \cos (\theta / 2) d \xi 
+ {\tilde A}_\eta \, \frac12 \sin (\theta / 2) d \eta. 
\label{eqn:topol-term}
\end{equation}
Thus we obtain 
\begin{equation}
\left\{
\begin{array}{l}
{\tilde A}_\theta 
= 2 A_4 \Bigl(\frac12 (\xi - \eta), \set{c_m} \Bigr),
\\
{\tilde A}_\xi = 2 A_0(\set{c_m}) \cos(\frac12 \theta) 
- 2 A_1 \Bigl(\frac12 (\xi - \eta), \set{c_m} \Bigr) 
\sin(\frac12 \theta),
\\
{\tilde A}_\eta = - 2 A_0(\set{c_m}) \sin(\frac12 \theta) 
- 2 A_1 \Bigl(\frac12 (\xi - \eta), \set{c_m} \Bigr) 
\cos(\frac12 \theta),
\end{array}
\right.
\label{eqn:fic-gauge-pots}
\end{equation}
which are evidently non-singular. 
\section{Complex Variable Parametrizations of the Spin CS}
\label{sec:c-parametCS}
The generic spin CS and CSPI in $\S$ \ref{sec:CS} 
and $\S$ \ref{sec:PI} can be put into complex variable forms 
like the conventional ones 
\cite{Rad,Pera,Arec}; 
The number of complex variables is, however, twice; This causes a need
for a supplementary condition to recover the proper degrees of freedom
of $\bf\Omega$.
These problems are discussed in 
$\S$ \ref{sec:compVar1}. 
Next, we employ the complex variable form to illustrate the contraction 
procedure from the generic spin CS and PI to 
the corresponding CCS and PI ($\S$ \ref{sec:high-spin}). 
Besides we add another  complex variable form ($\S$ \ref{sec:compVar2}). 
\subsection{Complex variable form via Gaussian decomposition}
\label{sec:compVar1}
We can parametrize the spin CS, $\ket{\bf\Omega}$, 
by a pair of complex variables 
${\bf z} \equiv (z_{+}, z_{-})$ 
and its complex conjugate 
${\bf z^*} \equiv (z_{+}^{*}, z_{-}^{*})$ 
via the ``Gaussian decomposition'' of the operator 
${\hat R}({\bf\Omega})$, 
i.e.(\ref{eqn:compRot})-(\ref{eqn:EulerComp}): 
\begin{equation}
\ket{\bf\Omega} = \ket{\bf z} 
\equiv {\hat R}({\bf z}) \ket{\Psi_0} 
\equiv {\hat R}(z_{+}, z_3, z_{-}) \ket{\Psi_0}. 
\label{eqn:def-z-spinCS}
\end{equation} 
We put ${\hat R}({\bf z}) 
=  {\hat R}(z_{+}, z_3, z_{-})$ 
since from (\ref{eqn:EulerComp}) $z_3$ is a function of 
$\bf z$: 
\begin{equation}
\exp(- z_3 / 2) 
= i \, \{ \, z_{+}^{*} z_{-}^{*} \, 
/ \, [ \, \absq{z_{+}} (1 + \absq{z_{+}} )^2 \, 
] \, \}^{1/2} .
\label{eqn:z3}
\end{equation}
However, we know that the degrees of freedom of the CS, 
i.e. those of $\bf\Omega$, are three; 
And thus it is clear that the representation by ${\bf z}$ 
and ${\bf z^*}$ is still redundant. 
We may remedy the problem by reducing the degrees of freedom 
with the aid of a subsidiary condition: 
$\abs{z_{+}} = \abs{z_{-}}$ from (\ref{eqn:EulerComp}). 
\par
For the resolution of unity we have
\begin{equation}
\int \ket{\bf z} d \nu({\bf z}) 
\bra{\bf z} = {\bf 1},
\label{eqn:comp-resolution}
\end{equation}
where
\begin{equation}
d \nu({\bf z}) 
= \frac{2 s + 1}{2 \pi^2} 
\frac{\delta(\abs{z_{+}} - \abs{z_{-}})}
{\abs{z_{+}}\, (1 + \abs{z_{+}}^2)^2} 
\, d^2(z_{+}) d^2(z_{-}) 
\label{eqn:z-inv-mes}
\end{equation}
and
$
d^2(z_\ell) 
\equiv 
d ({\rm Re} z_\ell) d ({\rm Im} z_\ell)\ (\ell = +, -).
$
\par
The propagator reads: 
\begin{equation} 
K({\bf z}_f, t_{f}; {\bf z}_i, t_{i})
=\int  
\exp \{ ({i} / \hbar) S[{\bf z}(t)] \} \, 
{\cal D} [{\bf z}(t)],
\label{eqn:zPI}
\end{equation}
where
\begin{equation}
S[{\bf z}(t)]
\equiv \int_{t_i}^{t_f}  
\Bigl[ \ \bra{\bf z} {i} \hbar 
\frac{\partial}{\partial t} 
 \ket{\bf z} 
 - H({\bf z}, t) \ \Bigr] 
 \, d t 
 \equiv \int_{t_i}^{t_f}   
 L({\bf z}, {\bf \dot z},t) \, d t 
\label{eqn:action-z}
\end{equation}
with 
\begin{equation}
H({\bf z}, t) 
\equiv 
\bracket{{\bf z}}{\hat H}{{\bf z}}
\qquad 
{\rm and} 
\qquad
{\cal D}
[{\bf z}(t)] 
\equiv
\lim_{N \rightarrow \infty} 
\prod_{j=1}^{N} 
d \nu({\bf z}_{t_j}).
\label{eqn:H+paths-z}
\end{equation}
The explicit form of the Lagrangian yields
\begin{eqnarray}
L({\bf z}, {\dot {\bf z}}, t) 
& = & i \hbar \, 
\Bigl\{ \, 
\frac{1}{2 \absq{z_+}} 
\Bigl[ \, 
A_0(\set{c_m}) \Bigl( \, \frac{1 - \absq{z_+}}{1 + \absq{z_+}} \, 
( z_+^{*} {\dot z_+} - {\dot z}_+^{*} z_+) 
+ (z_{-}^{*} {\dot z_-} - {\dot z}_-^{*} z_- ) \, 
\Bigr) \, 
 \Bigr] 
\nonumber \\
& & \qquad \qquad 
+ A_3({\bf z}, {\dot {\bf z}}; \set{c_m}) \, 
\Bigr\} 
 - H({\bf z}, t),
\label{eqn:zLagpol}
\end{eqnarray}
where 
\begin{equation}
A_3({\bf z}, {\dot {\bf z}}; \set{c_m}) 
\equiv 
\frac{1}{\absq{z_+} (1 + \absq{z_+})} 
\sum_{m=-s+1}^{s} f(s, m) 
(  c_m c_{m-1}^{*}  z_+ {\dot z_+}^* z_- 
- c.c. ).
\label{eqn:deff-zA3}
\end{equation}
One can obtain the above relations
(\ref{eqn:comp-resolution})-(\ref{eqn:deff-zA3}), 
via (\ref{eqn:def-z-spinCS})-(\ref{eqn:z3}) and 
(\ref{eqn:compRot})-(\ref{eqn:EulerComp}), 
from the Euler angle parametrization 
forms in $\S$ \ref{sec:CS} - $\S$ \ref{sec:PI}; 
Or one may confirm them by calculating 
$\bra{\bf z} (\partial/ \partial t) \ket{\bf z}$ directly. 
The formulae for the conventional CS, 
$
\ket{z} = (1 + \absq{z}){}^{- 1 / 2} \exp(z {\hat S}_{+})\ket{\Psi_0}
$ with $\ket{\Psi_0} = \ket{- s}$, follow by 
putting $z_{-} = - z_{+}^{*}$ and neglecting or integrating out 
$\abs{z_{-}}$-variable and then replacing $z_{+}$ with $z$. 
%
\subsection{High spin limit: 
contraction to the canonical CS}
\label{sec:high-spin}
It is well-known that in the high spin limit, 
i.e. $s \rightarrow \infty$, the conventional spin CS with 
$\ket{\Psi_0} = \ket{s}$ or $\ket{- s}$ approaches to the the usual CCS 
\cite{Rad,Arec,Perb}. 
However, as we saw in $\S$ \ref{sec:CCSPI}, 
the CCS and CCSPI has been extended to an arbitrary 
FV case; And thus, as we put in  $\S$ \ref{sec:intro}, 
it is natural to ask whether there exits any spin CS and CSPI 
that tends to the CCS and CCSPI with an arbitrary FV. 
This has been one of the motivations mentioned in $\S$ \ref{sec:intro} 
to construct such general spin CS and CSPI as described 
in $\S$ \ref{sec:CS} - $\S$ \ref{sec:compVar1}. 
The answer is affirmative:
\begin{theorem}
The spin CS $\ket{\bf\Omega}$ and CSPI with a generic FV in $\S$
 \ref{sec:CS} - $\S$ \ref{sec:compVar1}, 
in the  high spin limit, tend to the CCS $\ket{\alpha}$ and CCSPI 
described in $\S$ \ref{sec:CCSPI}. 
\end{theorem}
%
\begin{proof}
We adapt the method of Radcliffe \cite{Rad} and Arecchi 
{\em etal.} \cite{Arec} for a generic FV case.
Following the high spin limit of the transformation 
{\it {\`a} la} Holstein-Primakoff 
 \cite{HolPri}, let us put 
\begin{equation}
{\hat S}_+ \rightarrow (2 s)^{1 /2} {\hat a}^+, 
\quad
{\hat S}_- \rightarrow (2 s)^{1 /2} {\hat a}, 
\quad
{\hat S}_3 \rightarrow -  s\, {\bf 1} 
+ {\hat a}^+ {\hat a}.
\label{eqn:high-spin-S}
\end{equation}
We also set
\begin{equation}
z_+ \rightarrow \alpha 
(2 s)^{- 1 /2}, \qquad
z_- \rightarrow - z_+^*, 
\label{eqn:high-spin-z}
\end{equation}
which, with (\ref{eqn:z3}), gives 
\begin{equation}
z_3 \rightarrow \absq{\alpha} / (2 s). 
\label{eqn:high-spin-z3}
\end{equation}
\par
Then the combination of (\ref{eqn:high-spin-S})-(\ref{eqn:high-spin-z3})
 with (\ref{eqn:compRot})-(\ref{eqn:EulerComp}) 
and (\ref{eqn:CCS-D}) produces 
\begin{equation}
{\hat R}({\bf \Omega}) 
= {\hat R}({\bf z})
\longrightarrow 
\exp(\alpha {\hat a}^+) 
\exp \Bigl(- ( 1 / 2) \absq{\alpha} \Bigr) 
\exp(- \alpha^* {\hat a}) 
= {\hat D}(\alpha).
\label{eqn:spin2canon-ops}
\end{equation}
Besides since from (\ref{eqn:high-spin-S})
\begin{equation}
{\hat a}^+ {\hat a} \ket{n} = n \ket{n},  
\qquad \bigl( 
n \equiv m + s, \ket{n} 
\equiv \lim_{s \rightarrow \infty} \ket{m}
\bigr),
\label{eqn:spin2canon-baseket}
\end{equation}
we obtain
\begin{equation}
\ket{\Psi_0} = \sum_{m = - s}^{s} c_m \ket{m}
\longrightarrow  
\sum_{n = 0}^{\infty} c_n \ket{n},
\label{eqn:spin2canon-FV}
\end{equation}
where the numbering of the coefficients has been shifted.  
From (\ref{eqn:spin2canon-ops}) and 
(\ref{eqn:spin2canon-FV}) 
we obtain that
\begin{equation}
\ket{\bf\Omega}  = \ket{\bf z} 
= {\hat R}({\bf z}) \ket{\Psi_0}
\longrightarrow 
{\hat D}(\alpha) \cdot 
\sum_{n=0}^{\infty} c_{n} \ket{n} 
= \ket{\alpha},
\label{eqn:spin2canon-CS}
\end{equation}
which is precisely (\ref{eqn:CCS1}) in $\S$ \ref{sec:gCCS}: 
the definition of the CCS with a generic FV. 
\par
Next, with the aid of (\ref{eqn:z-inv-mes}), 
(\ref{eqn:high-spin-z}) and (\ref{eqn:spin2canon-CS}), 
we find that the left side hand of (\ref{eqn:comp-resolution}) becomes 
\begin{eqnarray}
\int\ket{\bf z} d \nu({\bf z}) 
\bra{\bf z} 
& = &  \frac{2 s + 1}{2 \pi^2} 
\int \ket{\bf z} \bra{\bf z} \cdot 
\frac{\delta(\abs{z_{+}} - \abs{z_{-}})} 
{\abs{z_{+}}\, (1 + \abs{z_{+}}^2)^2} 
\cdot 
\abs{z_{+}} \,  \abs{z_{-}} 
\nonumber \\
& & \qquad \qquad \times \, 
 d(\abs{z_{+}}) d({\rm arg} \, z_{+}) \, 
d(\abs{z_{-}}) d({\rm arg} \, z_{-})
\nonumber \\
& \rightarrow & 
\frac{2 s + 1}{2 \pi^2} 
\int \ket{\bf z} \bra{\bf z} 
\cdot 
\frac{\abs{z_{+}}}{(1 + \absq{z_{+}})^2} 
\cdot 2 \pi \cdot \delta \bigl( {\rm arg}\, z_{-} 
- ({\rm arg}\, z_{+} - \pi) \bigr) \,
\nonumber \\
& & \qquad \qquad \times \, 
 d(\abs{z_{+}}) d({\rm arg} \, z_{+}) \, 
d({\rm arg} \, z_{-})
\nonumber \\
& = & \frac{2 s + 1}{\pi} 
 \int \ket{z_+} \bra{z_+} 
\cdot 
\frac{\abs{z_{+}}}{(1 + \absq{z_{+}})^2} 
\, 
d(\abs{z_{+}}) d({\rm arg} \, z_{+})
\nonumber \\
& \rightarrow & 
\frac{1}{\pi} \frac{2 s + 1}{2 s} 
\int  \ket{\alpha} \bra{\alpha} 
\cdot 
\frac{\abs{\alpha}}
{[1 + \bigl(\absq{\alpha} / (2 s) \bigr)]^2}
\, 
d(\abs{\alpha}) d({\rm arg} \, \alpha) 
\nonumber \\
& \rightarrow & 
\frac{1}{\pi}
\int 
\ket{\alpha} d^2 \alpha \bra{\alpha}, 
\label{eqn:ResUni-spinCS2CCS}
\end{eqnarray}
which shows that the resolution of unity for 
the spin CS, (\ref{eqn:polresolution}) or (\ref{eqn:comp-resolution}), 
tends to that for CCS (\ref{eqn:CCSresolution}). 
The arguments in the $\delta$-function in (\ref{eqn:ResUni-spinCS2CCS}) 
should be interpreted as ``modulo $2\pi$''. 
\par
Now that we have both Eqs. (\ref{eqn:spin2canon-CS}) 
and (\ref{eqn:ResUni-spinCS2CCS}), 
we see that all the results of the spin CS and CSPI here 
approach to those in $\S$ \ref{sec:CCSPI}, which completes the proof. 
\end{proof}
\par
We may also see the results from the complex variable PI expression 
(\ref{eqn:zPI})-(\ref{eqn:deff-zA3}) 
with the help of (\ref{eqn:high-spin-z}) 
and (\ref{eqn:high-spin-z3}). 
To this end, notice that we have in $s \rightarrow \infty$ limit 
\begin{equation}
A_0(\set{c_m}) 
= \sum_{m = - s}^{m = s} m \absq{c_m}
= \sum_{m = - s }^{m = s} (n - s) \absq{c_m}
\longrightarrow 
- s 
\end{equation}
and 
\begin{equation}
f(s, m) \longrightarrow n^{1 / 2} (2 s)^{1 / 2}\ 
.
\end{equation}
Then we find that the Lagrangian (\ref{eqn:zLagpol}), 
which is equivalent to (\ref{eqn:Lagpol}), 
for the generic spin CSPI tends to (\ref{eqn:LagCCSPI}) for the CCSPI. 
And the results in $\S$ \ref{sec:CS} - $\S$ \ref{sec:compVar1} 
are converted to those in $\S$ \ref{sec:CCSPI}. 
Especially, we see that the $A_3$-term in (\ref{eqn:deffA3}) 
and (\ref{eqn:deff-zA3}) corresponds to the $A$-term in
(\ref{eqn:defA-CCS}). 
\par
The $A_3$-term is not represented as a total derivative; 
And thus the $A_1$- and $A_4$-terms in the $A_3$-term take part in variational equations (\ref{eqn:vareq-pol}) for the spin CS. It is merely in the high spin limit that the $A_3$-term, 
approaching to the $A$-term, becomes a total derivative 
and its effect disappears in the 
variational equations. Revisit $\S$ \ref{sec:CCSCE} and 
$\S$ \ref{sec:sc-spinCSPI} for the point. 
\subsection{Another complex variable form}
\label{sec:compVar2}
We have another complex variable representation 
of the CS \cite{Ino}.
To this end we write ${\hat R}^{(1/2)}(\bf\Omega)$ 
in (\ref{eqn:s1/2}), using a new pair of complex variables 
${\bf a} = (a_1, a_2)$, in the form of 
\begin{equation}
{\hat R}^{(1/2)}({\bf\Omega}) 
= {\hat R}^{(1/2)}({\bf a})
= \left(
\begin{array}{lr}
a_1 
& - a_2^{*}
\\
a_2
& a_1^{*}
\end{array}
\right)
\qquad 
{\rm with} 
\qquad 
\absq{a_1} + \absq{a_2} = 1, 
\label{eqn:complex2-a}
\end{equation}
which is often used for the SU(2) group. 
We see from (\ref{eqn:s1/2}) and (\ref{eqn:complex2-a})
\begin{equation}
a_1 = \cos(\theta / 2) \exp[- i (\phi + \psi) / 2], 
\quad
a_2 = \sin(\theta / 2) \exp[ i (\phi - \psi) / 2].
\label{eqn:complex2-Euler}
\end{equation}
\par
The spin CS, in this case, is specified by 
\begin{equation}
\ket{\bf a}
= {\hat R}({\bf a}) \ket{\Psi_0}
\equiv 
{\hat R}({\bf\Omega})  \ket{\Psi_0},
\label{eqn:def-a-spinCS}
\end{equation}
where 
$\bf a$ is related to $\bf \Omega$ via 
(\ref{eqn:complex2-Euler}). 
\par
The resolution of unity becomes 
\begin{equation}
\int \ket{\bf a} d \lambda({\bf a}) 
\bra{\bf a} = {\bf 1},
\label{eqn:a-resolution}
\end{equation}
where
\begin{equation}
d \lambda({\bf a}) 
= \frac{4(2 s + 1)}{\pi^2} 
\, 
\delta({\bf a}^2 - 1) d^2 {\bf a} 
\quad 
{\rm and} 
\quad 
d^2 {\bf a} 
\equiv 
d^2 a_1 d^2 a_2.
\label{eqn:a-inv-mes}
\end{equation}
with 
$
d^2(a_\ell) 
\equiv 
d ({\rm Re}\  a_\ell) \, d ({\rm Im} \ a_\ell)\ 
(\ell = 1, 2).
$ 
The $\delta$-function leaves the degrees of freedom 
being three as (\ref{eqn:z-inv-mes}). 
\par
The propagator reads:
\begin{equation} 
K({\bf a}_f, t_{f}; {\bf a}_i, t_{i})
=\int  
\exp \{ ({i} / \hbar) S[{\bf a}(t)] \} \, 
{\cal D} [{\bf a}(t)],
\label{eqn:aPI}
\end{equation}
where
\begin{equation}
S[{\bf a}(t)]
\equiv \int_{t_i}^{t_f}   
\Bigl[ \ \bra{\bf a} {i} \hbar 
\frac{\partial}{\partial t} 
 \ket{\bf a} 
 - H({\bf a}, t) \ \Bigr] \, d t 
 \equiv \int_{t_i}^{t_f}   
 L({\bf a}, {\bf \dot a},t) \, d t 
\label{eqn:action-a}
\end{equation}
with 
\begin{equation}
H({\bf a}, t) 
\equiv 
\bracket{{\bf a}}{\hat H}{{\bf a}}
\qquad 
{\rm and} 
\qquad
{\cal D}
[{\bf a}(t)] 
\equiv
\lim_{N \rightarrow \infty} 
\prod_{j=1}^{N} 
d \lambda({\bf a}_{t_j}).
\label{eqn:H+paths-a}
\end{equation}
The explicit form of the Lagrangian yields
\begin{eqnarray}
L({\bf a}, {\dot {\bf a}}, t) 
& = & i \hbar 
\Bigl[  
A_0(\set{c_m}) 
\Bigl( 
(a_1^{*} {\dot a}_1 - {\dot a}_1^{*} a_1) 
+ (a_2^{*} {\dot a}_2 - {\dot a}_2^{*} a_2) 
\Bigr) 
\nonumber \\
& & \qquad 
+ A_3({\bf a}, {\dot {\bf a}}; \set{c_m}) 
\Bigl] - H({\bf a}, t),
\label{eqn:Lag-a} 
\end{eqnarray}
where 
\begin{equation}
A_3({\bf a}, {\dot {\bf a}}; \set{c_m}) 
\equiv 
\sum_{m = - s + 1}^{s} f(s, m) 
[ c_m c_{m - 1}^{*} (a_1 {\dot a}_2 - {\dot a}_1 a_2)
- {\rm c.c.}].
\label{eqn:deff-aA3}
\end{equation}
\par
We may also put the results in a real variable form using ${\bf x}
\equiv (x_1, \cdots, x_4)$ via (\ref{eqn:complex2-a}) 
and (\ref{eqn:real-R}). 
In the case the restriction 
${\bf x}^2 = 1$ keeps the degree of freedom of the CS being three. 
\section{Summary and Prospects}
\label{sec:summary} 
We have investigated a natural extension 
of the spin or SU(2)CS and their PI forms by using arbitrary FV, 
which turns out to be performed successfully. 
\par
In the present paper we have worked on the basic formulation. 
The physical applications, 
in relation to fictitious monopoles and geometric phases, 
will be treated in subsequent papers separately. 
We will discuss criteria in choosing FV for real Lagrangians. 
The problem has a close link to that of the semiclassical 
versus full quantum evolutions of CS and FV. 
It was Stone \cite{Stone} who first raised the problem 
commenting on the previous version of our article 
\cite{spinPIGP}. 
He pointed out that an arbitrary FV is not always 
realized and that there may be restrictions on FV  
so that quantum evolutions are consistent with the 
semiclassical ones. 
The formal CSPI themselves do not give answers to it; 
And thus, one may ascribe the fault to the formal CSPI. 
We have resolved, in the present article, the mysteries
 posed in \cite{Stone} to some extent by proving 
the process from discrete CSPI to the continuous ones 
as described in $\S$ \ref{sec:d2cPI}. 
The fact that the spin CSPI in 
$\S$ \ref{sec:PI} - $\S$  \ref{sec:compVar1} 
certainly contract to the CCSPI in  $\S$ \ref{sec:CCSPI} strengthens the validity 
of the formulation.  
We will clarify the riddle more deeply the next time around. 
However, the whole problem seems to have rather 
profound nature and we will still need much further investigations.   
\par
Next, from a broader viewpoint 
let us put the prospects of the future below.
First, conventional CCS and spin CS 
have been playing the roles of macroscopic wave functions 
in vast fields from lasers, superradiance, superfluidity 
and superconductivitiy to nuclear and particle physics 
\cite{KlaSk}. 
CS have such potential. 
And thus we may expect that by choosing appropriate sets of 
$\set{c_m}$ the CS evolving from arbitrary FV 
will serve as approximate states or trial wave functions 
for the collective motions, having higher energies 
in various macroscopic or mesoscopic quantum phenomena 
such as spin vortices \cite{KraYab} and domain walls, 
which may not be treated by the former. 
We hope that numerous applications of the CS and CSPI 
will be found in the near future. 
Second, from the viewpoint of mathematical physics 
as well as physical applications, 
it is desirable that the present CS and CSPI formalism 
is extended to wider classes. 
The generalization to the SU(1, 1) CS case, 
 which is closely related to squeezed states 
in lightwave communications and quantum detections 
 \cite{SQ,KMb}, 
is one of the highly probable candidates. 
We also have another candidate, i.e. the SU(3)CS case. 
Remembering that the SU(2)CS with a general FV here 
extends the SU(2) BMS${}^2$ Lagragngian and 
gives a clear insight into the topological terms, 
the SU(3)CS case may also shed a new light 
on the original SU(3) Wess-Zumino term 
\cite{WZ,Wit}. 
Finally, as we put in $\S$ \ref{sec:intro}, we may regard CS 
with arbitrary FV as quantum states without classical analogues. 
We have already known some of such states 
\cite{DNS,AAM-C}. 
It is true that CS with the conventional FV are closest to classical states 
and have useful properties 
\cite{Perc}. 
However, since ``the physical world is quantum mechanical'' \cite{Fey}, 
it seems definitely right to search boldly new quantum 
states whether their classical counterparts exist or not. 
\section*{Acknowledgments}
The author is grateful to Prof. M. Stone for sending the draft of \cite{Stone}. 
For other acknowledgments, 
see those in the first version of the present article 
 \cite{spinCSPI3a-v1}.
\appendix
\setcounter{equation}{0}
\renewcommand{\theequation}{%
                  \Alph{section}.\arabic{equation}}
\renewcommand{\thesection}{%
                  {\bf Appendix} \Alph{section}}
\section{Some Formulae for Rotation Matrices}
\label{sec:Rots}
Some basic formulae on the properties of the rotation matrices 
are enumerated 
\cite{Mess,BS,VMK,Miller}. 
We employ them in $\S$ \ref{sec:CS} - $\S$ \ref{sec:c-parametCS}. 
We mainly follow the notation and convention of Messiah \cite{Mess}.
\\

\noindent
 (i) {\em Matrix elements}\\
A rotation with Euler angles ${\bf \Omega} \equiv (\phi, \theta, \psi)$ 
of a spin-$s$ particle is specified by an operator 
$
{\hat R}({\bf \Omega}) 
= \exp(-i \phi {\hat S_3}) \exp(-i \theta {\hat S_2}) 
\exp(-i \psi {\hat S_3})
$; It has a $(2 s + 1) \times (2 s + 1)$ matrix representation 
whose $(m, m')$-entry is 
\begin{equation}
{\sf R}_{m{}m'}^{(s)}({\bf \Omega}) 
\equiv \bra{m} {\hat R}({\bf \Omega}) \ket{m'} 
= \exp(-i \phi m) \, {\sf r}_{m{}m'}^{(s)}(\theta) \, 
\exp(-i \psi m').
\label{eqn:mat-R}
\end{equation}
Here ${\sf r}_{m{}m'}^{(s)}(\theta) \equiv \bra{m} 
\exp(-i \theta {\hat S_2}) \ket{m'}$ 
is determined by the formula due to Majorana 
\cite{Maj} and to Wigner \cite{Mess}:
\begin{equation}
{\sf r}_{m{}m'}^{(s)}(\theta) 
= \sum_{t} 
N(s, m, m'; t) \cdot [\cos(\theta / 2)]^{2s + m - m' - 2t} \cdot 
[\sin(\theta / 2)]^{2 t - m + m'}
\label{eqn:mat-r}
\end{equation}
with
\begin{equation}
N(s, m, m'; t) \equiv (- 1)^{t} 
\frac {[\, (s + m)! \, (s - m)! \, (s + m')! \, (s - m')! \, ]^{1/2}}
{(s+ m - t)! \, (s - m'- t)! \, t! \, (t - m + m')!},
\label{eqn:N(s, m, m'; t)}
\end{equation}
where the sum runs over any integer $t$ by which all the factorials in
(\ref{eqn:N(s, m, m'; t)}) make sense. 
In particular, if $s=\frac12$, 
${\sf r}_{m{}m'}^{(s)}$ is extremely simple to give: 
\begin{equation}
{\hat R}^{(1/2)}({\bf \Omega}) 
= \left(
\begin{array}{lc}
\cos( \frac12 \theta ) 
\exp[- \frac12 i (\phi + \psi)] 
& -\sin( \frac12 \theta ) 
\exp[- \frac12  i (\phi - \psi)]
\\
\sin( \frac12 \theta ) \exp[ \frac12 i (\phi - \psi)] 
& \cos( \frac12 \theta ) \exp[ \frac12 i (\phi + \psi)] 
\end{array}
\right).
\label{eqn:s1/2}
\end{equation}
Most of the following relations, being independent of $s$, 
can be readily verified by the use of (\ref{eqn:s1/2}).  
For a higher spin $s$ one can find explicit expressions of 
${\sf r}_{m{}m'}^{(s)}$  in \cite{BS} 
and \cite{VMK}.
\\
(ii) {\em Group manifold}\\
Introducing the real variables ${\bf x} \equiv (x_1, \cdots, x_4)$
via 
\begin{equation}
{\hat R}^{(1/2)}({\bf \Omega}) 
= \left(
\begin{array}{lc}
x_1 + i\,  x_2 
& - x_3 +i\,  x_4
\\
x_3 +i\,  x_4
& x_1 - i\,  x_2
\end{array}
\right)
\equiv {\hat R}^{(1/2)}({\bf x}),
\label{eqn:real-R}
\end{equation}
we see that ${\bf x}^2 = x_1^2 + x_2^2 + x_3^2 + x_4^2 = 1$ 
results; 
Thus the $SO(3)$ manifold is isomorphic to a 3-sphere $S^3$. 
\\
(iii) {\em Gaussian decomposition} \cite{Arec,Perb,Perc,Gil}\\
The rotation matrix ${\hat R}({\bf\Omega})$ can be 
put into the normal or anti-normal ordering form 
in which ${\hat R}$ is 
specified by a set of complex variables:
\begin{eqnarray}
{\hat R}({\bf \Omega}) 
= {\hat R}(z_+, z_3, z_-)
& \equiv & \exp(z_+ {\hat S_+}) \exp(z_3 {\hat S_3}) 
\exp(z_- {\hat S_-})
\nonumber \\
& = & \exp(z_- {\hat S_-}) \exp(- z_3 {\hat S_3}) 
\exp(z_+ {\hat S_+}) 
\label{eqn:compRot}
\end{eqnarray}
The relation between the Euler angles and the complex parameters 
is given by
\begin{equation}
\left\{
\begin{array}{l}
z_+ = - \tan ( \frac12 \theta ) \exp(- i \phi) 
 \\
z_3 = - 2 \ln \{ \cos( \frac12 \theta ) 
\exp[ \frac12 i (\phi + \psi)] \} \\
z_- = \tan ( \frac12 \theta ) \exp(- i \psi).
\end{array}
\right.
\label{eqn:EulerComp}
\end{equation} 
(iv) {\em Combinations with $\bf S$} \cite{Mess}
\begin{equation}
\left\{
\begin{array}{l}
{\hat R^{+}}({\bf\Omega}) {\hat S}_3 {\hat R}({\bf\Omega}) 
 = \cos\theta {\hat S}_3 
 - \frac12 \sin\theta 
 [ \exp(i \psi) {\hat S}_{+} + \exp(-i \psi) {\hat S}_{-} ]
 \\
{\hat R}^{+}({\bf\Omega}) {\hat S}_{\pm} 
{\hat R}({\bf\Omega})  =  \exp(\pm i \phi) 
\{ 
\sin\theta  {\hat S}_3 
+ \frac12 
[
(\cos\theta \pm 1) \exp(i \psi) 
{\hat S}_{+}  \\
 \qquad \qquad \qquad + (\cos\theta \mp 1) \exp(- i \psi) 
{\hat S}_{-}
]
\}. 
\end{array}
\right.
\label{eqn:RS}
\end{equation}
(v) {\em Inverse}\\
$\hat R({\bf\Omega})$ is unitary and its inverse matrix is given by 
\begin{equation}
{\hat R}^{+}(\phi, \theta, \psi) 
= {\hat R}^{-1}(\phi, \theta, \psi)
= {\hat R}(-\psi, -\theta, -\phi). 
\label{eqn:rotinverse}
\end{equation} 
(vi) {\em Orthogonality relation}\\
The relation stems from integrating the products 
of the unitary irreducible representations of a compact 
group over the element of the group; 
thus it is a generic relation 
for the representations. In the present case it reads 
\cite{BS,Miller}:
\begin{equation}
\int_0^{2\pi} \int_0^{\pi} \int_0^{2\pi}\ 
\left( {\sf R}_{m{}m'}^{(s)}({\bf \Omega}) \right)^{*} 
{\sf R}_{n{}n'}^{(s')}({\bf \Omega}) 
\sin\theta\ d \phi d \theta d \psi 
= \frac{8\pi^{2}}{2s+1}\ 
\delta_{m, n} \delta_{m', n'} \delta_{s, s'}.
\label{eqn:ortho}
\end{equation}
(vii) {\em Two successive rotations}\\
Two successive rotations specified by Euler angles  
${\bf\Omega}_\ell \equiv (\phi_\ell, \theta_\ell, \psi_\ell) 
\  (\ell = 1, 2)$ 
produce 
$
{\hat R}({\tilde {\bf\Omega}}) 
\equiv 
{\hat R}({\bf\Omega}_2) {\hat R}({\bf\Omega}_1), 
$ 
where ${\tilde {\bf\Omega}} \equiv ({\tilde \phi}, {\tilde \theta}, 
{\tilde \psi})$ obeys 
\begin{equation}
\left\{
\begin{array}{l}
\cos{\tilde \theta} = \cos\theta_1 \cos\theta_2 
- \sin\theta_1 \sin\theta_2 \cos(\phi_1 + \psi_2) 
 \\
\sin{\tilde \theta} \exp(i {\tilde \phi}) 
\\
=  \exp(i \phi_2) \, \Bigl[ \cos\theta_1 \sin\theta_2 
+ \sin\theta_1 \cos\theta_2 \cos(\phi_1 + \psi_2) 
\\ \qquad \qquad  \qquad \quad
+ i \sin\theta_1 \sin(\phi_1 + \psi_2) \Bigr] 
\\
\cos( \frac12 {\tilde \theta} ) 
\exp[ \frac12 i (\tilde \phi + \tilde \psi)] 
\\
 =  \exp[ \frac12 i (\phi_2 + \psi_1) ] 
 \Bigl\{ 
\cos( \frac12 \theta_1 ) \cos( \frac12 \theta_2) 
\exp[ \frac12 i (\phi_1 + \psi_2) ] 
\\
 \qquad 
\qquad \qquad 
- \sin( \frac12 \theta_1 ) 
\sin( \frac12 \theta_2 ) 
\exp[ -\frac12 i (\phi_1 + \psi_2) ] \Bigr\}.
\end{array}
\right.
\label{eqn:tworots}
\end{equation}
(viii) {\em Three successive rotations}\\
In a similar manner to that in (vii), the Euler angles made of 
three successive rotations can be calculated. 
Assuming that the rotations are specified by Euler angles 
$(\phi_1, \theta_1, \psi_1)$, $(\phi, \theta, \psi)$ 
and $(\phi_2, \theta_2, \psi_2)$, which happen in this order, 
the composed rotation yields 
$
{\hat R}({\bf\Omega}') 
\equiv 
{\hat R}({\bf\Omega}_2)
{\hat R}({\bf\Omega})
{\hat R}({\bf\Omega}_1), 
$
where ${\bf\Omega}' \equiv (\phi', \theta', \psi')$ obeys  
\begin{eqnarray}
\cos{\theta'} 
& = & [ \cos\theta_1 \cos\theta 
- \sin\theta_1 \sin\theta \cos(\phi_1 + \psi) 
]\cos\theta_2 
\nonumber \\
& & \quad 
+ \{ \sin\theta_1 [ \sin(\phi_1 + \psi) 
\sin(\phi + \psi_2) 
- \cos(\phi_1 + \psi) \cos\theta \cos(\phi + \psi_2) ] 
\nonumber \\ 
& & \qquad \quad 
- \cos\theta_1 \sin\theta \cos(\phi + \psi_2) 
\} \sin\theta_2
\label{eqn:trirots}
\end{eqnarray}
and two additional equations that we omit here; 
They describe $\sin\theta' \exp(i \phi')$ 
and $\cos(\frac12 \theta')$ \\
$\times \exp[\frac12 i (\phi' + \psi')]$ in terms of 
${\bf\Omega}_1$, ${\bf\Omega}$ 
and ${\bf\Omega}_2$ as in (\ref{eqn:tworots}).

\end{document}